\documentclass[11pt]{article}
\usepackage[margin=1in]{geometry}
\usepackage{amsmath,amssymb,amsthm,amsxtra}
\usepackage{mathrsfs,mathtools}
\usepackage{bbm}
\usepackage{bm}
\usepackage[shortlabels]{enumitem}
\usepackage{algorithm}
\usepackage{algpseudocode}
\usepackage{float}
\usepackage{microtype}
\microtypesetup{expansion=false}
\usepackage[T1]{fontenc}
\usepackage[utf8]{inputenc}
\usepackage{xcolor}
\usepackage{hyperref}
\usepackage[nameinlink,noabbrev,capitalise]{cleveref}
\makeatletter
\renewcommand{\theHALG@line}{\thealgorithm.\arabic{ALG@line}}
\makeatother

\hypersetup{hidelinks}
\numberwithin{equation}{section}

\newtheorem{theorem}{Theorem}[section]
\newtheorem{lemma}[theorem]{Lemma}

\newtheorem{corollary}[theorem]{Corollary}

\newtheorem{definition}[theorem]{Definition}

\newcommand{\BIS}{\#\mathrm{BIS}}
\newcommand{\TV}{\operatorname{\mathsf{TV}}}

\newcommand{\Id}{\mathrm{Id}}

\newcommand{\cP}{\mathcal{P}}

\newcommand{\cQ}{\mathcal{Q}}

\newcommand{\R}{\mathbb{R}}

\newcommand{\E}{\operatorname{\mathbb{E}}}

\newcommand{\norm}[1]{\left\lVert#1\right\rVert}

\DeclareMathOperator*{\argmin}{argmin}

\newcommand{\vol}{\operatorname{vol}}

\newcommand{\one}{\mathbf 1}

\title{The C-3PO Framework: Algorithms for $\#\mathrm{BIS}$-Hard Problems on Graphs with Low Threshold Rank or High Girth}
\author{
Ijay Narang\thanks{School of Computer Science, Georgia Institute of Technology,
\texttt{inarang3@gatech.edu}}
}
\date{}
\begin{document}
\maketitle
\begin{abstract}
We give approximate counting and sampling algorithms for spin models on graphs of low threshold rank or high girth. Given any $n$-vertex graph $G$, let $k_\eta  := k_\eta(G)$ count normalized-Laplacian eigenvalues strictly below $\eta$. For $0<\varepsilon<1/2$, an $\varepsilon$-approximate counting algorithm outputs an estimate $\widehat Z$ of the partition function $Z$ satisfying
$(1-\varepsilon)Z\le\widehat Z\le(1+\varepsilon)Z$ with probability at least $3/4$.
An $\varepsilon$-approximate sampling algorithm outputs a configuration whose distribution is within total variation distance $\varepsilon$ of the target distribution. We give the following results:

\begin{itemize}
    \item \emph{Ferromagnetic Potts.}
    Fix $q\ge2$ and $\eta\in(0,1]$. There exists $\beta_0(q,\eta)$ such that, for every fixed inverse temperature
    $\beta\ge\beta_0$, the zero-field $q$-state Potts model admits approximate counting and sampling algorithms:
    \begin{itemize}
        \item On arbitrary graphs in time $\exp(O(k_\eta))\operatorname{poly}(n,1/\varepsilon).$
        \item Fix $\Delta \ge 3$. On $\Delta$-regular graphs of girth $g$ in time $\exp(O(n e^{-g})) \operatorname{poly}(n,1/\varepsilon).$
    \end{itemize}
    \item \emph{Bipartite hard-core.}
    Fix $\Delta\ge3$ and $\eta\in(0,1]$. There exists $\lambda_0(\Delta,\eta)$ such that, for every fixed fugacity $\lambda\ge\lambda_0$, the hard-core model admits approximate counting and sampling:
    \begin{itemize}
        \item On $\Delta$-regular vertex-transitive bipartite graphs in time $\exp(O (k_\eta))\operatorname{poly}(n,1/\varepsilon).$
        \item On $\Delta$-regular bipartite graphs of girth $g$ in $\exp(O(n e^{- g}))\operatorname{poly}(n,1/\varepsilon).$
    \end{itemize}
\end{itemize}

All four results follow from verifying the hypotheses of a new framework we call C-3PO. At a high level, the C-3PO algorithm proceeds by enumerating small cuts and reconstructing the original measure through cluster expansion around each small cut. Underlying the algorithm is an improvement of subspace enumeration methods via active regression.

\end{abstract}

\clearpage

\section{Introduction}\label{sec:c3po-introduction}

A spin model on a graph $G=(V,E)$ assigns to each vertex a spin
from a finite set $[q]$. A configuration $z\in[q]^V$ has weight
\[
    w_G(z)
    =\prod_{v\in V} b_v(z_v)
     \prod_{uv\in E} A_{uv}(z_u,z_v),
\]
where the nonnegative vertex weights $b_v$ and edge interactions
$A_{uv}$ specify the model. The associated partition function and Gibbs distribution of a spin model are
\[
    Z_G=\sum_{z\in[q]^V}w_G(z),
    \qquad
    \mu_G(z)=\frac{w_G(z)}{Z_G}.
\]
Spin models have received much attention in statistical physics and theoretical computer science, as they describe interacting particles and encode combinatorial structures. Two fundamental computational tasks are computing the partition function and sampling from the Gibbs distribution. Exact evaluation of the partition function is $\#\mathrm{P}$-hard in general~\cite{Valiant79}, motivating the study
of approximation algorithms. 

For $0<\varepsilon<1/2$, a randomized,  $\varepsilon$-approximate counting
algorithm outputs a number $\widehat Z$ satisfying
$(1-\varepsilon)Z_G\le\widehat Z\le(1+\varepsilon)Z_G$
with probability at least $3/4$. An $\varepsilon$-approximate sampling
algorithm outputs a configuration with distribution $\widehat\mu$
satisfying $\TV(\widehat\mu,\mu_G)\le\varepsilon$.

In this paper, we study approximate counting and sampling for ``cut-based'' spin models. Let $\mathcal F_G$ denote the set of feasible configurations (where feasibility is not necessarily dependent on cuts). Cut-based models are models which take the form
\[
    Z_G(\theta)=\sum_{z\in\mathcal F_G}e^{-\theta c(z)},\qquad \mu(z)= \frac{e^{-\theta c(z)}}{Z_G(\theta)}, 
    \qquad c(z)=|\{uv\in E:z_u\ne z_v\}|.
\]
We call the parameter \(\theta\) the inverse temperature, since large \(\theta\) corresponds to low temperature, as the model concentrates around configurations with small cuts. Note that this formulation encompasses a broad class of models such as the zero-field ferromagnetic Ising and Potts models with uniform interactions, and the hard-core model on regular bipartite graphs under the encoding given in
Section~\ref{sec:hardcore-application}.

Two prominent approaches to approximate counting and sampling exploit different temperature regimes. At high temperature, weak interactions often allow local Markov chains, such as Glauber dynamics, to mix rapidly. At low temperature, the Gibbs distribution can instead concentrate around competing ordered phases. Moving between them through local updates may require passing through configurations of very small probability, creating bottlenecks and causing exponentially slow mixing \cite{BCFKTVV99, MWW09}. 

Under suitable structural assumptions on the graph, polymer models and cluster expansions provide an alternative in the low temperature regime~\cite{HPR20,JKP20,JPP23,CHP26,NP}. These approaches exploit additional structure (such as expansion)
to identify a small explicit family of well-separated dominant
configurations and represent deviations around each as collections of defects, called polymers. Examples include the monochromatic configurations for the low temperature ferromagnetic Potts model on lattice tori~\cite{HPR20}, and the two configurations occupying either bipartition class for the high-fugacity hard-core model on bipartite expanders~\cite{JKP20,JPP23,NP}.

When these defects are sufficiently suppressed, a convergent cluster expansion allows the contribution of each phase to be approximated efficiently. In these settings, the corresponding polymer models represent nearly disjoint regions of the configuration space, allowing their contributions to be combined with almost no error from overlap, yielding efficient approximate counting and sampling algorithms.

On general graphs, this need not be the case. There may be exponentially many competing reference configurations, and their polymer expansions may overlap substantially. Our main contribution is a method of dealing with this more general case. In particular, our framework uses an improved subspace enumeration algorithm to construct a collection of reference configurations with small cuts. These references serve as the bases for polymer expansions that together capture almost all Gibbs mass. By combining their contributions (and carefully  correcting for overlaps), we obtain algorithms for approximate counting and sampling in cut-based models.

For notation, let $G$ be an $n$-vertex graph, $A_G$ and $D_G$ its adjacency and degree matrices and $L_G=\Id-D_G^{-1/2}A_GD_G^{-1/2}$ its normalized Laplacian.
Write $d_v=\deg_G(v)$ and $\partial T=\{uv\in E:|\{u,v\}\cap T|=1\}$ for the degree of $v$ and the edge boundary of $T\subseteq V$, respectively.
For a fixed $\eta\in(0,1]$, define its threshold rank by
\[
    k_\eta(G)=|\{i:\lambda_i(L_G)<\eta\}|.
\]
Let $\mathcal G_{k_\eta}$ denote the class of graphs with threshold rank at most $k_\eta$. We note that though low threshold rank is necessary for efficient subspace-enumeration algorithms~\cite{KT07,Kolla11,abs}, it is not too restrictive of a property, and holds for broad classes of graphs. For example, spectral expanders have $k_\eta(G)=1$, and high girth also gives low threshold rank (see Lemma~\ref{lem:girth-rank}). Additionally, dense graphs have low threshold rank, indeed, for fixed $0<\eta<1$, a graph of minimum degree $\delta$ satisfies $k_\eta(G)\le \frac{n}{\delta(1-\eta)^2}.$

Informally, our main result is as follows.

\begin{theorem}[Theorem \ref{thm:structural-rank}, informal]
Consider a spin model on $G$ whose feasible configurations $z$ have weights
\[
    w(z) \propto e^{-\theta c(z)},
    \qquad
    c(z)=|\{uv\in E(G):z_u\ne z_v\}|.
\]
If the following conditions are met:
\begin{itemize}
    \item (Threshold Rank) The graph $G$ has minimum degree $\delta = \delta(G)$ and threshold rank satisfying $k_\eta(G)\le k_\eta$.
    \item (Local Symmetry) There exists a grouping of vertex--spin pairs $(v,i)$ into types such that pairs of the same type have matching radius-$R$ neighborhoods up to relabeling vertices and spins.
\end{itemize}
Then, for every sufficiently large fixed $\theta$ and every
$0<\varepsilon<1/2$, there are $\varepsilon$-approximate counting and
sampling algorithms on the class $\mathcal G_{k_\eta}$ with running time
\[
 \exp\!\left(O\!\left(k_\eta+n e^{-\delta R}\right)\right)
 \operatorname{poly}(n,1/\varepsilon).
\]
\end{theorem}

We remark that both the threshold-rank bound and symmetry conditions are critical to our algorithm. At a high level, the threshold-rank assumption enables us to find the dominant cut configurations efficiently, while the symmetry condition enables us to control their contributions using cluster expansion. Before we discuss the algorithm and proof underlying our main result, we demonstrate its utility through two examples.

\subsection{Application I: The Ferromagnetic Potts Model}
\label{sec:intro-potts}

The zero-field ferromagnetic $q$-state Potts model assigns a spin in
$[q]$ to each vertex and rewards agreement across edges. For inverse
temperature $\beta\ge0$, its partition function and Gibbs distribution
are
\[
    Z_G(q,\beta)=\sum_{z\in[q]^V}e^{\beta m_G(z)},
    \qquad
    \mu_{G,q,\beta}(z)=\frac{e^{\beta m_G(z)}}{Z_G(q,\beta)},
    \qquad
    m_G(z)=\sum_{uv\in E}\mathbf1\{z_u=z_v\}.
\]
Thus, large inverse temperature corresponds to the low temperature regime considered by our framework. 

The approximation complexity differs sharply between two and more
than two spins. For $q=2$, the model is equivalent to the ferromagnetic
Ising model, for which Jerrum and Sinclair~\cite{JS} gave an FPRAS
on arbitrary graphs at every temperature. For $q\ge3$, Galanis,
\v{S}tefankovi\v{c}, Vigoda, and Yang~\cite{GSVY} proved that
approximating the partition function on bipartite graphs of maximum
degree at most $\Delta\ge3$ is $\BIS$-hard whenever
\[
    \beta>\beta_{\mathrm o}(q,\Delta)
    :=\log\!\left(
        \frac{q-2}{(q-1)^{1-2/\Delta}-1}\right).
\]
where $\beta_{\mathrm o}$ is referred to as the order--disorder threshold, at which
the ordered and disordered phases have equal limiting pressure on
random regular graphs. Expansion supports algorithms in this difficult
regime. Using polymer models around the constant colorings, Jenssen, Keevash, and Perkins~\cite{JKP20} gave an FPTAS and
an efficient sampler at low temperature on bounded-degree expanders. Carlson, Davies, Fraiman, Kolla, Potukuchi, and Yap~\cite{CDFKPY} later sharpened the sufficient temperature and expansion conditions on regular graphs, obtaining algorithms closer to the order--disorder threshold. Carlson, Davies, and Kolla~\cite{CDK24} gave related low temperature counting algorithms under a small-set expansion condition.

Our first result on the ferromagnetic Potts model gives algorithmic guarantees in terms of threshold rank alone.

\begin{theorem}\label{thm:potts-application}
Fix $q\ge2$, $\eta\in(0,1]$. There exists
$\beta_0=\beta_0(q,\eta)$ such that, for every fixed
$\beta\ge\beta_0$ and every $0<\varepsilon<1/2$, the zero-field ferromagnetic $q$-state Potts
model admits $\varepsilon$-approximate counting and sampling algorithms on the
class $\mathcal G_{k_\eta}$ in time
\[
    \exp(O(k_\eta))
    \operatorname{poly}(n,1/\varepsilon).
\]
\end{theorem}

We remark that, to the best of our knowledge, Theorem~\ref{thm:potts-application} gives the most general graph conditions for counting and sampling low-temperature ferromagnetic Potts. Indeed, previous work required additional assumptions such as bounded degree or small set expansion~\cite{CDK24}. 

Additionally, our algorithms naturally extend  to $\Delta$-regular, high girth graphs. Let $\mathcal{G}_{\Delta}^g$ denote the class of $\Delta$-regular graphs with girth $g$.

\begin{theorem}\label{thm:girth-potts}
Fix $q\ge2$ and $\Delta\ge3$, and let $\theta_0$ be the threshold from Theorem~\ref{thm:structural-rank}.
For every fixed $\beta\ge\theta_0(q,1/100,0)$ and every $0<\varepsilon<1/2$, the zero-field
ferromagnetic $q$-state Potts model admits $\varepsilon$-approximate
counting and sampling algorithms on the class $\mathcal G_\Delta^g$ in time
\[
\exp\!\left(O(n e^{- g})\right) \operatorname{poly}(n,1/\varepsilon).
\]
\end{theorem}

\subsection{Application II: The Bipartite Hard-core Model}
\label{sec:intro-hardcore}

The hard-core model is a distribution on the independent sets
$\mathcal I(G)$ of a graph $G$. Given a fugacity $\lambda>0$, its
partition function and Gibbs distribution are
\[
    Z_G(\lambda)=\sum_{I\in\mathcal I(G)}\lambda^{|I|},
    \qquad
    \mu_{G,\lambda}(I)=\frac{\lambda^{|I|}}{Z_G(\lambda)}.
\]
Let $\mathcal{G}^{\leq}_\Delta$ denote the set of all graphs of maximum
degree at most $\Delta$. On this class of graphs, the hard-core model exhibits a computational threshold at $\lambda_c(\Delta)$.  
When $\lambda>\lambda_c(\Delta)$, there is no efficient approximate sampler
or FPRAS on $\mathcal{G}^{\leq}_\Delta$ unless
$\mathrm{NP}=\mathrm{RP}$~\cite{Sly10,SS14,gsvhc}; when
$\lambda<\lambda_c(\Delta)$, efficient approximate counting and sampling
are possible through correlation decay and the self-avoiding-walk tree~\cite{Weitz06}. The value $\lambda_c(\Delta)$ corresponds to the uniqueness threshold on the infinite $\Delta$-regular tree, and is explicitly known to be  $\lambda_c(\Delta) = \frac{(\Delta - 1)^{\Delta - 1}}{(\Delta - 2)^{\Delta}}$.

For bipartite graphs, however, the situation is less well understood.
At $\lambda=1$, approximately counting independent sets is the canonical problem $\BIS$, which is neither
known to admit an FPRAS nor known to be NP-hard to
approximate~\cite{DGGJ04,DGJ10}.
Moreover, many natural approximate counting problems have been shown
to be $\BIS$-hard~\cite{DGGJ04,DGJ10}, making the complexity of $\BIS$
a central open problem in approximate counting. Additionally, let $\mathcal{G}^{\leq}_{\Delta,\Delta}$
denote the class of bipartite graphs with maximum degree at most
$\Delta$. Cai, Galanis, Goldberg, Guo, Jerrum,
{\v{S}}tefankovi{\v{c}}, and Vigoda~\cite{cai} showed that, for every
fixed $\Delta\ge3$ and $\lambda>\lambda_c(\Delta)$, approximating
$Z_G(\lambda)$ on graphs $G\in\mathcal{G}^{\leq}_{\Delta,\Delta}$
is $\BIS$-hard. This hardness also extends to the non-uniqueness regime on $\mathcal{G}_{\Delta, \Delta}$, the class of $\Delta$-regular bipartite graphs \cite{narang2026separations}.

Nevertheless, progress on approximate counting and sampling beyond
uniqueness has been made for special classes of bipartite graphs,
including lattices~\cite{HPR20, CHP26}, bipartite
expanders~\cite{JKP20,CGGPSV21,JPP23,NP}, and random regular bipartite
graphs~\cite{LLLM22,CGSV22,KOGT26}. Such algorithms proceed by observing that at low temperature, independent sets corresponding to an
entire bipartition class give explicit reference configurations, with deviations from each described by a polymer model. Most recently, Narang and Perkins~\cite{NP} refined the high-fugacity
polymer analysis to give algorithms under the spectral expansion
condition $\sigma_2(M_G)\le(1-\eta)\Delta$. For a $\Delta$-regular bipartite graph with bi-adjacency matrix $M_G$, the bound $\sigma_2(M_G)\le(1-\eta)\Delta$ is equivalent to $k_\eta(G)=1$. Our first hard-core result can thus be viewed as an extension of polymer models for bipartite graphs beyond the second largest eigenvalue.

Recall that a graph automorphism of $G=(V,E)$ is a bijection $\phi:V\to V$ such that $\{u,v\}\in E$ if and only if $\{\phi(u),\phi(v)\}\in E$.
We say that $G$ is vertex-transitive if, for every $u,v\in V$,
there exists an automorphism $\phi$ of $G$ satisfying $\phi(u)=v$. Let $\mathcal{G}_{\Delta, \Delta}^{\mathrm{vt}}(k_\eta)$ denote the set of $\Delta$-regular, vertex-transitive bipartite graphs with threshold rank at most $k_\eta$.

\begin{theorem}\label{thm:vt-hardcore}
Fix $\Delta\ge3$, $\eta\in(0,1]$. There exists
$\lambda_0=\lambda_0(\Delta,\eta)$ such that, for every fixed
$\lambda\ge\lambda_0$ and every $0<\varepsilon<1/2$, the hard-core model admits
$\varepsilon$-approximate counting and sampling algorithms on the class
$\mathcal{G}_{\Delta,\Delta}^{\mathrm{vt}}(k_\eta)$ in time
\[
    \exp(O(k_\eta))
    \operatorname{poly}(n,1/\varepsilon).
\]
\end{theorem}

We remark that, in contrast to our Potts model results, our hard-core application requires vertex transitivity because of the asymmetry
between spins. However, we can also get algorithms without such a strong notion of symmetry. 

Indeed, our second hard-core result gives algorithms only dependent on girth. This is because in a regular graph of sufficiently large girth, neighborhoods of a given radius are identical rooted trees. These local symmetries (together with the fact that large girth gives good threshold rank) allow us to remove the vertex transitivity condition.

For notation, let $\mathcal{G}_{\Delta, \Delta}^g$ denote the set of all $\Delta$-regular bipartite graphs of girth $g$. 
\begin{theorem}\label{thm:girth-hardcore}
Fix $\Delta\ge3$. For every fixed
$\lambda\ge\exp\!\left(\Delta\theta_0(2,1/100,2/\Delta)\right)$ and every
$0<\varepsilon<1/2$, the hard-core model admits $\varepsilon$-approximate counting and sampling algorithms on the class $\mathcal{G}_{\Delta,\Delta}^g$ in time
\[
 \exp\!\left(O(n e^{- g})\right)
 \operatorname{poly}(n,1/\varepsilon).
\]
\end{theorem}

\section{The C-3PO Framework}
\label{sec:techniques}

Our algorithms are based on the observation that at low temperatures, ``cut-based'' models place most of their Gibbs mass near small cuts, and that on graphs with low threshold rank, we can efficiently find such cuts via subspace-enumeration based methods. Thus, the C-3PO method proceeds in three stages. We begin by constructing a list of all small reference cuts ($C^3$ stage). We then utilize a polymer model to approximate local deviations around each reference cut ($P$ stage), and then finally correct for overlaps ($O$ stage) through rejection sampling. 

We begin by defining the local symmetry and type balance that our algorithms require. Write $B_G(v,R)$ for the vertices within distance $R$ of $v$, or its connected component when $R=\infty$.

\begin{definition}[Local symmetry and type balance]
\label{def:local-symmetry}
Let $R\in\mathbb Z_{\ge0}\cup\{\infty\}$ and $B\ge0$.
A pairwise $q$-spin model on $G=(V,E)$ has local symmetry
and type balance $(R,B)$ if there are a finite type set
$\mathcal T$ and a map $t:V\times[q]\to\mathcal T$ such that:
\begin{enumerate}[label=(\alph*),ref=\thetheorem(\alph*)]
    \item\label{cond:3a} \emph{(Local symmetry)}
    Whenever $t(v,i)=t(w,j)$, there are an isomorphism
    $\pi:G[B_G(v,R)]\to G[B_G(w,R)]$ preserving degrees in $G$
    and a permutation $\sigma$ of $[q]$ with
    $\pi(v)=w$ and $\sigma(i)=j$.
    For each edge $xy$ in the first ball and $a,b\in[q]$,
    $(a,b)$ is permitted at $(x,y)$ if and only if
    $(\sigma(a),\sigma(b))$ is permitted at $(\pi(x),\pi(y))$.

    \item \emph{(Type balance)}
    Write $n_\tau(z)=|\{v\in V:t(v,z_v)=\tau\}|$
    and fix an ordering of each edge's endpoints. There exists a constant configuration $z_0$ and functions
    $a_{uv,\tau}:[q]^2\to\R$ such that:
    \begin{enumerate}[label=(\roman*),ref=\theenumi(\roman*)]
        \item\label{cond:3b} Every feasible configuration $z$ satisfies
        \[
            n_\tau(z)-n_\tau(z_0)
            =\sum_{uv\in E}a_{uv,\tau}(z_u,z_v)
            \qquad(\tau\in\mathcal T).
        \]
        \item\label{cond:3c} For every edge $uv$ and permitted spin pair $(i,j)$,
        \[
            \sum_{\tau\in\mathcal T}|a_{uv,\tau}(i,j)|
            \le B\,\mathbf1\{i\ne j\}.
        \]
    \end{enumerate}
\end{enumerate}
\end{definition}

Roughly speaking, the above conditions state that for any vertex--spin pairs of the same type, their radius-$R$ neighborhoods look identical after a relabeling of vertices
and spins. The reason for this condition is that while we utilize subspace enumeration to identify low energy regions, the fluctuations around each cut may vary significantly. This local symmetry prevents that from happening. We remark that despite its technical nature, the above condition is quite easy to satisfy, as we will see in the proofs of Theorems \ref{thm:potts-application}, \ref{thm:girth-potts}, \ref{thm:vt-hardcore}, and \ref{thm:girth-hardcore}.

Indeed, we note that the zero-field ferromagnetic Potts model satisfies these conditions on every graph with $R=\infty$ and $B=0$. Take $\mathcal T=V$ and $t(v,i)=v$, so that all spins at a given
vertex have the same type. The identity map on vertices, together with a spin permutation sending $i$ to $j$, gives the required local symmetry. Moreover, under this choice, each type occurs exactly once in every configuration, and so $n_\tau(z)=n_\tau(z_0)=1$ for every $\tau\in\mathcal T$. Taking every $a_{uv,\tau}=0$ therefore gives type balance.

We may now state our main  result.

\begin{theorem}\label{thm:structural-rank}
Fix $q\ge2$, $\eta\in(0,1]$, and $B\ge0$.
There exists $\theta_0=\theta_0(q,\eta,B)$ such that, for every
$k_\eta\ge1$, $R\in\mathbb Z_{\ge0}\cup\{\infty\}$,
and fixed $\theta\ge\theta_0$,
there are $\varepsilon$-approximate counting and sampling
algorithms for every $\varepsilon > 0$ on graphs $G \in \mathcal G_{k_\eta}$ for pairwise
cut-based $q$-spin models satisfying the following conditions.
\begin{enumerate}
    \item\label{cond:1}
    For every $i\in[q]$, the constant configuration
    $i\cdot\mathbf1_n$ is feasible.

    \item\label{cond:2}
    For every feasible configuration $z$, every $i\in[q]$,
    and every connected component $C$ of $G[z^{-1}(i)]$,
    recoloring all vertices of $C$ to any spin appearing
    on a neighbor outside $C$ preserves feasibility.

    \item \label{cond:3} The spin model on $G$ satisfies Definition~\ref{def:local-symmetry} with parameters $(R,B)$.
\end{enumerate}
Furthermore, the run-time of our approximation algorithms is
\[
    \exp\!\left(O\!\left(k_\eta+n e^{-\delta R}\right)\right)
    \operatorname{poly}(n,1/\varepsilon),
\]
where $n$ and $\delta$ are the number of vertices and minimum
degree of the input graph, respectively, and
$e^{-\delta R}=0$ when $R=\infty$.
\end{theorem}

Before we present the inner workings of the \textbf{$C^3$}, \textbf{$P$}, and \textbf{$O$} steps, we  describe the high-level ideas gluing them all together.

\subsection{Overview of the Technical Ideas} \label{subsec:technical overview}
Let us begin by briefly recalling how polymer models are utilized in the algorithms of \cite{JKP20, HPR20, JPP23, CHP26, NP} for $\#\mathrm{BIS}$. In the regimes considered by these algorithms, expansion (either combinatorial or spectral) ensures that almost all Gibbs mass comes from independent sets described as deviations around one of the two configurations occupying an entire side of the bipartition $V=V_L\sqcup V_R$. Thus, one may approximate the partition function by
\[
Z_G(\lambda) \approx \lambda^{|V_L|} \Xi_{V_L} + \lambda^{|V_R|} \Xi_{V_R}
\]
where $\lambda^{|V_L|}$ and $\lambda^{|V_R|}$ are the weights of the two reference configurations, and $\Xi_{V_L}$ and $\Xi_{V_R}$ account for the deviations around these configurations.

We generalize this approach to cut-based models. Let $b$ be a cut budget, specified below. At low temperature, we aim to approximate $Z_G(\theta)$ using the sum
\[
\sum_{h \in \mathcal H_b} e^{-\theta c(h)} \Xi_h,
\]
We therefore need two ingredients: a collection of small-cut references whose polymer models represent all but a small fraction of the Gibbs mass, and an efficient way to count and sample the deviations around each reference. 

For the latter, we must choose a polymer model so that the weights decay rapidly enough in the degree-volume of the vertices they change to compensate for the number of possible polymers (and thus ensure convergence of the cluster expansion). In our model, a change from $h$ to $z$ has relative weight $\exp\!\bigl(-\theta(c(z)-c(h))\bigr).$ Thus, a natural way to obtain this decay is therefore to require each polymer to increase the cut by an amount proportional to the degree-volume of the vertices it changes. 

Such conditions, known as Peierls conditions, are classical in statistical physics~\cite{Peierls36,friedli2017statistical} and also appear in algorithmic work~\cite{CGGPSV21}. We shall use the following Peierls-type condition.

Set $\alpha=\eta/(16(q-1))$ and write
$\operatorname{vol}(T)=\sum_{v\in T}d_v$.

\begin{definition}[Peierls-type condition]\label{def:peierls}
Let $h,z\in\mathcal F_G$ and let
$T=\{v\in V:h_v\ne z_v\}$.
We say that the change from $h$ to $z$ satisfies a Peierls condition with constant $\alpha$ if
\[
    c(z)-c(h)\ge\alpha\operatorname{vol}(T).
\]
\end{definition}

We may now state the polymer model that we shall utilize around each reference.

\begin{definition}[Polymer model around a reference]
\label{def:overview-polymer-model}
Fix a feasible reference $h\in\mathcal F_G$.
A polymer around $h$ is a pair $\gamma=(T,\sigma)$, with
$T\subseteq V$ and $\sigma\in[q]^T$, satisfying the following
conditions. Let $h^\gamma$ denote the configuration that agrees with $\sigma$ on $T$ and with $h$ outside $T$.
\begin{enumerate}[label=(\roman*)]
    \item $T$ is nonempty and connected.

    \item $\sigma_v\ne h_v$ for every $v\in T$, and
    $h^\gamma\in\mathcal F_G$.

    \item The change from $h$ to $h^\gamma$ satisfies the
    Peierls-type condition in Definition~\ref{def:peierls}, namely
    \begin{equation}\label{eq:overview-polymer-bound}
        \Delta_h(\gamma):=c(h^\gamma)-c(h)
        \ge\alpha\operatorname{vol}(T).
    \end{equation}
\end{enumerate}
The weight of $\gamma$ is $w_h(\gamma)=e^{-\theta\Delta_h(\gamma)}$.
Additionally, we say that two polymers around $h$ are compatible if their supports are disjoint and no edge joins them.
The polymer partition function and Gibbs distribution of the polymer model are then
\[
    \Xi_h=\sum_{\Gamma\text{ compatible}}
        \prod_{\gamma\in\Gamma}w_h(\gamma),
    \qquad
    \nu_h(\Gamma)=\frac{1}{\Xi_h}
        \prod_{\gamma\in\Gamma}w_h(\gamma),
\]
In the above, applying $\Gamma$ to $h$ means replacing the spins on its polymer
supports by their assigned spins, leaving all other spins unchanged.
\end{definition}

As is standard, we establish convergence of the above model by comparing the polymer weights with the number of possible polymers. Definition \ref{def:peierls} bounds the weight of a polymer of degree-volume $m$ by $e^{-\theta\alpha m}$. To bound the number of such polymers, we count connected supports from a fixed vertex, utilizing a branching argument to include each newly encountered vertex. As a support of degree-volume $m$ requires at most $m$ binary decisions, and admits at most $(q-1)^m$ spin assignments, there are at most $[2(q-1)]^m$ such polymers that contain any fixed vertex.

Thus, for sufficiently large $\theta$, the weight decay dominates this growth:
\[
    \sum_{\gamma=(T,\sigma):\,v\in T}
    w_h(\gamma)e^{2\operatorname{vol}(T)}
    \le
    \sum_{m\ge1}
    \bigl[2(q-1)e^{-(\theta\alpha-2)}\bigr]^m
    \le \frac12.
\]
Every polymer incompatible with a polymer supported on $T$ contains a vertex in its closed neighborhood, which has size at most $2\operatorname{vol}(T)$. Summing the preceding bound over this neighborhood allows us to verify the Koteck\'y--Preiss criterion (Lemma \ref{lem:marked-KP}), giving absolute convergence of the cluster expansion for $\log\Xi_h$.

From the above, we now have a convergent polymer model around any feasible reference $h$. The key point now, is to choose a collection of references so that their polymer models collectively represent every feasible configuration. Given any configuration $z$, we want to assign it to a reference $h$ such that $z$ will be included in the polymer model around $h$. Equivalently, every connected component on which $h$ and $z$ disagree defines a valid polymer. To achieve this (while simultaneously favoring references with small cuts), we consider the following function:

\begin{equation}\label{eq:overview-reference-map}
    \kappa(z)\in\operatorname*{arg\,min}_{h\in\mathcal F_G}
    \left\{
        c(h)+\alpha\sum_{v:h_v\ne z_v}d_v
    \right\},
\end{equation}
breaking ties by a fixed ordering. The map $\kappa$ thus assigns a reference to each feasible configuration. We call a configuration $h$
\emph{stable} if every monochromatic component $C$ of $h$ satisfies
\begin{equation}\label{eq:overview-stability}
    |\partial C|\le\frac{\eta}{16}\operatorname{vol}(C).
\end{equation}

\begin{lemma}[Properties of $\kappa(z)$]\label{lem:overview-reference-properties}
\label{lem:c3po-canonical}
Assume Condition~\ref{cond:2}. For every $z\in\mathcal F_G$, let $h=\kappa(z)$.
Then:
\begin{enumerate}
    \item For every connected component $T$ of $\{v:h_v\ne z_v\}$,
    the pair $\gamma_T=(T,z|_T)$ is a polymer around $h$. In particular,
    $h^{\gamma_T}\in\mathcal F_G$ and
    \[
        c(h^{\gamma_T})-c(h)\ge\alpha\operatorname{vol}(T).
    \]
    \item The reference $h$ is stable.
\end{enumerate}
\end{lemma}

As we require our reference cuts to be both stable and small, let $b$ be a cut budget (to be chosen in \eqref{eq:overview-cutoff}). The list of references we seek is therefore
\begin{equation}\label{eq:overview-reference-family}
    \mathcal H_b=\{h\in\mathcal F_G:
        h\text{ is stable and }c(h)\le b\}.
\end{equation}
Our algorithm will thus proceed by enumerating $\mathcal H_b$ and computing the corresponding cluster expansions around each reference $h \in \mathcal H_b$. We will then finish by combining these polymer expansions by a rejection sampling procedure which corrects for overlaps arising when the same configuration $z$ is represented in both $\Xi_h$ and $\Xi_{h'}$.

To show that this algorithm yields a good approximation, we must show that restricting to references with small cuts keeps almost all the Gibbs mass. Condition~\ref{cond:1} supplies a feasible constant reference $z_0$ with $c(z_0)=0$, providing a baseline relative to which the factor $e^{-\theta c(h)}$ suppresses references with large cuts. However, the contribution around $h$ also includes $\Xi_h$, so we must ensure that the contribution of deviations does not dominate. It is here that local symmetry plays a crucial role. Under Condition~\ref{cond:3}, the polymer models around $h$ and $z_0$ have similar local contributions away from the cut edges, and so we can control the ratio $\frac{\Xi_h}{\Xi_{z_0}}$ in terms of $c(h)$. 
The resulting bound \eqref{eq:overview-partition-comparison} ensures that, for sufficiently large $\theta$, the contribution of deviations cannot overwhelm the penalty for a large cut. Thus, the polymer models around references with small cuts capture almost all Gibbs mass and yield a good approximation to $Z_G(\theta)$.

\subsection{Refined Subspace Enumeration and the Algorithmic Details}

A high-level outline of our algorithm is as follows.

\begin{algorithm}[H]
\caption{The C-3PO framework}
\label{alg:c3po-overview}
\begin{algorithmic}[1]
\Require A graph $G$, a cut-weighted spin model satisfying
Theorem~\ref{thm:structural-rank}, a sufficiently large fixed penalty
$\theta$, and accuracy $\varepsilon$.
\State \textbf{$C^3$: Constructing a Family of Reference Cuts.} Use refined
subspace enumeration to construct the set $\mathcal H_b$ of all feasible stable references with cut size at most the cutoff $b$ in \eqref{eq:overview-cutoff}. At sufficiently large penalty, Conditions~\ref{cond:1} and~\ref{cond:3} ensure that their deviations capture all but $O(\varepsilon)$ of the Gibbs mass.
\State \textbf{$P$: Polymer Model Approximations.} For every $h \in \mathcal H_b$, construct a polymer model that represents deviations around $h$. Utilize cluster expansion to estimate its partition function $\Xi_h$
and sample its deviations. Set $W_h=e^{-\theta c(h)}\Xi_h$.
\State \textbf{$O$: Correcting overlaps.} Propose a cut with probability proportional to $W_h$, sample its deviations, and
obtain a spin configuration $z$. Fix an ordering of $\mathcal H_b$.
Accept precisely when $h$ is the first reference in this order whose
polymer model represents $z$.
\State For sampling, return the first accepted configuration.
For counting, estimate the acceptance probability $p$ and return
$\bigl(\sum_h W_h\bigr)p$, using the computed estimates.
\end{algorithmic}
\end{algorithm}

\paragraph{The $C^3$ stage.}
We enumerate the references in $\mathcal H_b$ defined by \eqref{eq:overview-reference-family}. Though threshold rank does not provide these references directly, it places their degree-weighted indicator vectors close to a common low-dimensional spectral subspace. For notation, let $\ell_i=\lambda_i(L_G)$ for $i\in[n]$,
with orthonormal eigenvectors $v_1,\ldots,v_n$.
Put
\[
    S=\operatorname{span}\{v_i:\ell_i<\eta/2\},
    \qquad r=\dim S\le k_\eta(G),
\]
and let $\Pi_S$ be the orthogonal projector onto $S$. Given a $q$-spin configuration $h$, encode it by
\[
    Y_h^a=D_G^{1/2}\mathbf1_{\{h=a\}}\quad(a\in[q]),\qquad
    Y_h=[\,Y_h^1\ \cdots\ Y_h^q\,]\in\R^{n\times q},
\]
Then
\begin{align*}
    \sum_{a=1}^q (Y_h^a)^{\mathsf T}L_GY_h^a
    &=\sum_{uv\in E}\sum_{a=1}^q
      \bigl(\mathbf1\{h_u=a\}-\mathbf1\{h_v=a\}\bigr)^2=2c(h).
\end{align*}
Thus, with $\|\cdot\|_F$ denoting the Frobenius norm,
\begin{align*}
    \|(\Id-\Pi_S)Y_h\|_F^2 =\sum_{a=1}^q\sum_{i:\ell_i\ge\eta/2} |\langle Y_h^a,v_i\rangle|^2 \le\frac2\eta\sum_{a=1}^q (Y_h^a)^{\mathsf T}L_GY_h^a =\frac{4c(h)}\eta.
\end{align*}
Every small-cut configuration therefore has its degree-weighted indicator vectors close to the same $r$-dimensional subspace. Classical subspace-enumeration algorithms construct a net over the subspace and round its points to obtain discrete candidates~\cite{KT07,Kolla11,abs}. A sufficiently fine net gives error controlled by the cut size, as in \cite[Lemma~9]{CDKP24}, but enumerating it can require $n^{O(r)}$ time. We instead improve upon subspace enumeration by replacing the net step by the response-oblivious active-regression procedure of Chen and Price~\cite{CP}.
\\\\
\emph{Subspace Enumeration via Active Regression.} Before we elaborate
upon how we utilize active regression, let us briefly explain the
problem setting.

Let $y$ be an unknown vector
lying close to $S$. Active regression asks, given only $S$ and
oracle access to the coordinates of $y$, can we find a prediction
with error comparable to the optimal least-squares error using
as few queries as possible? The algorithm of Chen and Price~\cite{CP}
selects and weights $O(r)$ coordinates using only $S$ (without
knowing $y$). Given the responses to the $O(r)$ queries of $y$, it computes
a prediction $\widehat y\in S$ satisfying (with constant probability)
\[
    \|\widehat y-y\|_2^2
    =O\!\left(\min_{u\in S}\|u-y\|_2^2\right).
\]

To apply this to our setting, note that although we do not know $h$, each queried vertex has only $q$ possible spin values. We apply regression to each $Y_h^a$, using the same queried vertices and weights for every color. Since the queried vertices are selected independently of $h$, we can try all $q^{O(r)}$ combinations of spin values at these vertices. Thus, each combination gives predicted indicator vectors
$\widehat Y^1,\ldots,\widehat Y^q$. From these, we can obtain $q^{O(r)}$ candidate configurations $x\in[q]^V$ by setting
\[
    x_v\in\operatorname*{arg\,max}_{a\in[q]}(\widehat Y^a)_v,
\]
and breaking ties arbitrarily. For a small cut configuration $h$, one combination agrees with $h$ at every queried vertex. Thus, with constant probability, its predictions give a candidate satisfying
\begin{equation}\label{eq:overview-candidate-error}
    \sum_{v:x_v\ne h_v}d_v
    =O\!\left(\sum_{a=1}^q\|(\Id-\Pi_S)Y_h^a\|_2^2\right)
    =O\!\left(\frac{c(h)}\eta\right).
\end{equation}
Repeating the query selection independently polynomially many
times and collecting the resulting candidates yields a family of
$q^{O(r)}\operatorname{poly}(n)$ candidates which are close to all the required references with high probability. Denote this collection of candidates by $\mathcal{Q}$.
\\\\
Since $\mathcal Q$ contains a nearby candidate for every small-cut configuration, it contains one for every $h\in\mathcal H_b$. It remains to recover these references exactly.
\\\\
\emph{Combinatorial Search for Stable References.}
Fix a reference $h\in\mathcal H_b$ and a nearby candidate $x$.
By \eqref{eq:overview-candidate-error}, the disagreement set
$D=\{v:x_v\ne h_v\}$ has degree-volume $O(b/\eta)$.
However, we do not know which vertices belong to $D$, and trying
all possible sets directly would take $n^{O(b/\eta)}$ time.
We will therefore work with the connected components of $D$.

This is where stability becomes useful. It lets us find vertices meeting the monochromatic components of each reference.
Indeed, \eqref{eq:overview-stability} gives each such component
a small boundary relative to its volume, placing its degree-weighted
indicator close to $S$. This allows us to use spectral sampling to find vertices meeting these components.

Specifically, sampling a vertex $v$ with probability $(\Pi_S)_{vv}/r$ hits each monochromatic component with probability
$\Omega(1/r)$. Furthermore, as we can also use stability and threshold rank to bound the number of these components by $r$, sampling $r$ vertices at a time and repeating this procedure gives a family of $\exp(O(r))\operatorname{poly}(n)$ sets. With high probability, every reference has a set in this family meeting all of its monochromatic components. 

We then use these sampled vertices and the candidates' cut edges to guide a search over connected changes of total degree-volume
$O(b/\eta)$. After pruning candidates with large cuts, this search
recovers $\mathcal H_b$ by retaining only feasible stable configurations
with cut at most $b$. Including active regression and spectral
sampling, the total running time of the full $C^3$ stage is thus
\[
    \exp(O(k_\eta(G)+b))
    \operatorname{poly}(n,\log(1/\zeta)),
\]
with success probability at least $1-\zeta$.

\paragraph{The $P$ stage.}
We shall now consider polymer models around each reference $h\in\mathcal H_b$. As discussed in Section~\ref{subsec:technical overview}, the Peierls-type condition ensures that a polymer supported on $T$ has weight at most $e^{-\theta\alpha\operatorname{vol}(T)}$. Since the number of polymers of a given degree-volume containing a fixed vertex grows at most exponentially in that volume, this decay ensures convergence of the cluster expansion for sufficiently large $\theta$. Thus, truncating these expansions gives efficient sampling and approximation of $\Xi_h$.

It remains to choose the cutoff $b$ to ensure that $Z_G(\theta) \approx \sum_{h \in \mathcal{H}_b} e^{-\theta c(h)} \Xi_h$. Although $e^{-\theta c(h)}$ penalizes references with large cuts, their polymer contributions $\Xi_h$ might offset this penalty. We shall therefore utilize the local symmetry provided by Condition \ref{cond:3} and Definition \ref{def:local-symmetry} to compare $\Xi_h$  with the contribution around the constant reference $z_0$ (which is feasible by Condition \ref{cond:1}). Let us now briefly discuss how we achieve this.
\\\\
\emph{Controlling Polymer Mass through Local Symmetry.}
Let $\Phi_h(\boldsymbol\gamma)$ denote the term corresponding to a cluster $\boldsymbol\gamma=(\gamma_1,\ldots,\gamma_k)$ in the cluster expansion of $\log\Xi_h$, and let $S_\gamma$ denote the support of a polymer $\gamma$. Additionally, define
\[
    p_v(h)
    :=\sum_{\boldsymbol\gamma}
      \frac{|\{i:v\in S_{\gamma_i}\}|}
           {\sum_{i=1}^k|S_{\gamma_i}|}
      \,\Phi_h(\boldsymbol\gamma),
\]
the portion of $\log\Xi_h$ attributed to the vertex $v$. By definition, we have $\log\Xi_h=\sum_{v\in V}p_v(h).$
Additionally, for each vertex $v$, let $\boldsymbol{h_v}=h_v\mathbf1_n$ denote the constant reference with spin $h_v$. We shall compare the cluster expansions through the following chain of inequalities:
\begin{align}
    |\log\Xi_h-\log\Xi_{z_0}|
    &\le \underbrace{
        \left|\log\Xi_h-\sum_v p_v(\boldsymbol{h_v})\right|
        }_{\text{(a)}} +\underbrace{
        \left|\sum_v p_v(\boldsymbol{h_v})-\log\Xi_{z_0}\right|
        }_{\text{(b)}}\notag\\
    &\le (q+1)c(h)
       +\frac12\sum_{\tau\in\mathcal T}
          |n_\tau(h)-n_\tau(z_0)|+E_R\notag\\
    &\le \left(q+1+\frac B2\right)c(h)+E_R.\label{eq:overview-partition-comparison}
\end{align}
where $E_R=2ne^{-\delta R}$ and $E_\infty=0$. In the above, the first inequality follows by the triangle inequality and the third is by the type balance condition of Definition \ref{def:local-symmetry}. The difficulty is in proving the second inequality.

To bound (a), observe that a cluster with no cut edges of $h$ included looks constant, and thus its contribution agrees with that around the constant reference $\boldsymbol{h_v}$ by Definition~\ref{def:overview-polymer-model}. Thus, only clusters meeting the endpoints of cut edges remain in the difference. As there are at most $2c(h)$ such endpoints, we can utilize \eqref{eq:volume-root-bound} to control the contribution through each one. Applying this bound to $h$ and the $q$ constant references gives an upper bound of $(q+1)c(h)$.

To bound (b), we use the local symmetry provided by Condition~\ref{cond:3a} to compare the constant-reference contributions at vertex--spin pairs of the same type. For clusters within the symmetry radius, the corresponding terms agree by local symmetry, and so the difference in their total contributions comes only from the changes in their type counts, which, by Conditions~\ref{cond:3b} and~\ref{cond:3c} are bounded by $Bc(h)$. For clusters extending beyond the symmetry radius, exponential decay in \eqref{eq:volume-root-bound} bounds their total contribution by $E_R$. Putting these together yields the desired upper bound.
\\\\
\emph{Choosing the cut budget.}
The comparison above establishes that most Gibbs mass is placed around references that correspond to small cuts. We now combine the reference-count bound from the $C^3$ stage with our comparison of polymer contributions. Since the polymer model around the zero-cut reference $z_0$ represents a subset of the feasible configurations, we have $Z_G(\theta)\ge\Xi_{z_0}$. Writing $C_{\rm press}=q+1+B/2$, equation~\eqref{eq:overview-partition-comparison} gives
\[
\begin{aligned}
    \mu\{z:c(\kappa(z))>b\}
    &\le \sum_{\substack{h\text{ stable}\\c(h)>b}}
        \frac{e^{-\theta c(h)}\Xi_h}{Z_G(\theta)} \le e^{E_R}\sum_{\substack{h\text{ stable}\\c(h)>b}}
        e^{-(\theta-C_{\rm press})c(h)}.
\end{aligned}
\]
Because threshold rank lets us bound the number of stable cuts (as in the $C^3$ stage), it follows that for sufficiently large fixed $\theta$, this tail decays exponentially in $b$. In particular, we can choose
\begin{equation}\label{eq:overview-cutoff}
    b=O\!\left(k_\eta(G)+E_R+\log(n/\varepsilon)\right)
\end{equation}
to make the omitted mass at most $O(\varepsilon)$.

\paragraph{The $O$ stage.}
We now have a collection of polymer models that together represent almost all Gibbs mass. It remains to combine them. The difficulty is that the same configuration may be represented around several references, so simply adding their contributions would count it more than once. To correct for this, we adapt the enumeration procedure of \cite{KLM89} to our setting.

Let
$\mathcal U_h$ be the configurations represented around $h$, and write
$\mathcal U=\bigcup_{h\in\mathcal H_b}\mathcal U_h$. Since
$\mu(\mathcal U^c)=O(\varepsilon)$, it suffices to approximate the total
weight in $\mathcal U$ and sample from $\mu(\cdot\mid\mathcal U)$.

Consider the distribution resulting from choosing $h$ with
probability proportional to $e^{-\theta c(h)}\Xi_h$ and then sampling $z$ from
the polymer model to produce the tuple $(h,z)$. Writing $\mu_W$ for this distribution and $W=\sum_{h\in\mathcal H_b}e^{-\theta c(h)}\Xi_h$ for the normalizing constant gives

\begin{equation}\label{eq:overview-representation-weight}
\begin{aligned}
    \mu_W((h,z)) = \mu_W(h)\,\mu_W(z\mid h) = \frac{e^{-\theta c(h)}\Xi_h}{W} \cdot 
       \frac{e^{-\theta(c(z)-c(h))}}{\Xi_h} = \frac{e^{-\theta c(z)}}{W}.
\end{aligned}
\end{equation}
Thus, the resulting probability of proposing $z$
is multiplied by the number of references representing it.

We correct this by keeping just one representation of each
configuration. Fix an ordering of $\mathcal H_b$ and accept a proposal $(h,z)$ only when its chosen reference is the first in this order whose polymer model $\Xi_h$ contains $z$. We check this efficiently using the conditions in Definition~\ref{def:overview-polymer-model}. Under this rule, exactly one proposal pair remains for each
$z\in\mathcal U$ and so conditioning on acceptance gives $\mu(\cdot\mid\mathcal U)$.

\section{Applications}\label{sec:applications}

We will now prove the four application theorems stated in the introduction by
verifying the conditions of Theorem~\ref{thm:structural-rank}.

\subsection{Ferromagnetic Potts model (Proof of Theorems \ref{thm:potts-application} and \ref{thm:girth-potts})} \label{subsec:ferro_apps}

\begin{proof}[Proof of Theorem~\ref{thm:potts-application}]

Rewrite
\[
    Z_G(q,\beta)
    =e^{\beta|E|}
      \sum_{z\in[q]^V}e^{-\beta c(z)},
    \qquad c(z)=|\{uv\in E:z_u\ne z_v\}|.
\]
so that the normalized model has the required form. We now verify the conditions of Theorem~\ref{thm:structural-rank}. For Conditions 1 and 2, every constant configuration is feasible,
and recoloring any monochromatic component preserves feasibility,
since there are no forbidden spin pairs.

For Condition 3, take $\mathcal T=V$ and $t(v,i)=v$.
To verify \ref{cond:3a}, suppose $t(v,i)=t(w,j)$, so $v=w$.
Choose the identity graph automorphism and any spin permutation
$\sigma$ with $\sigma(i)=j$. A global spin permutation preserves
every Potts interaction because
\[
    \mathbf{1}\{\sigma(a)=\sigma(b)\}
    =\mathbf{1}\{a=b\}.
\]
and thus Condition \ref{cond:3a} holds with $R=\infty$. For \ref{cond:3b}, every type $\tau\in V$ has
$n_\tau(z)=1$ for every configuration $z$. Choose
$z_0=\mathbf{1}_n$ and set $a_{uv,\tau}(i,j)=0$ for every
edge, type, and spin pair. Then
\[
    n_\tau(z)-n_\tau(z_0)
    =0=\sum_{uv\in E}a_{uv,\tau}(z_u,z_v).
\]
and so Condition \ref{cond:3c} holds with $B=0$. Take $\beta_0=\theta_0(q,\eta,0)$ and apply Theorem~\ref{thm:structural-rank} with $R=\infty$. Multiplying the partition-function estimate by $e^{\beta|E|}$ gives the desired result.
\end{proof}

To show Theorem \ref{thm:girth-potts}, we shall utilize the following lemma, which allows us to bound threshold-rank in terms of girth.

\begin{lemma}\label{lem:girth-rank} Let $G$ be a $\Delta$-regular graph of girth $g$, with $\Delta\ge3$.
Set $\rho_\Delta=2\sqrt{\Delta-1}/\Delta$. For $0<\eta<1-\rho_\Delta$, write
$a=\operatorname{arcosh}((1-\eta)/\rho_\Delta)$ and set
$a_\Delta:=\operatorname{arcosh}(99/(100\rho_\Delta))$.
Then
\[
 k_\eta(G)\le4n e^{-a(g-2)}.
\]
\end{lemma}
\begin{proof}
Put $\mathsf A=A_G/\Delta$ and let $\mathsf A_{\mathbb T}$ be normalized adjacency
on the infinite $\Delta$-regular tree, with adjacency operator $A_{\mathbb T}=\Delta\mathsf A_{\mathbb T}$. The spectrum of $\mathsf A_{\mathbb T}$ lies in
$[-\rho_\Delta,\rho_\Delta]$. Indeed, root the tree at $o$ and set
$w(v)=(\Delta-1)^{-\operatorname{dist}(o,v)/2}$. The neighbor sum
satisfies $(A_{\mathbb T}w)(v)/w(v)\le2\sqrt{\Delta-1}$, including
at the root. Summing
\[
 2|f(u)f(v)|\le\frac{w(v)}{w(u)}|f(u)|^2
                 +\frac{w(u)}{w(v)}|f(v)|^2
\]
over edges bounds the absolute quadratic form on finitely supported
$f$; density gives the claimed operator-norm bound.

For every $j<g$, closed walks of length $j$ based at a vertex of
$G$ lift to closed walks on its universal covering tree. Otherwise,
cancelling immediate edge reversals leaves a nonempty closed
nonbacktracking walk, which contains a cycle of length at most $j$.
Thus $(\mathsf A^j)_{vv}=(\mathsf A_{\mathbb T}^j)_{oo}$ for every $v$.
Let $m=\lfloor(g-1)/2\rfloor$ and
$Q(x)=T_m(x/\rho_\Delta)^2$, where $T_m$ is the Chebyshev polynomial.
Its degree is $2m<g$, it is nonnegative on the real line, and it is
at most one on the tree spectrum. With $\mathbf1_o$ denoting the unit vector at $o$, moment matching gives
\[
 \operatorname{tr}Q(\mathsf A)
 =n\langle\mathbf1_o,Q(\mathsf A_{\mathbb T})\mathbf1_o\rangle\le n.
\]
Every eigenvalue of $\mathsf A$ greater than $1-\eta$ contributes at least
$\cosh^2(ma)$. Hence
\begin{align*}
 k_\eta(G)&\le n/\cosh^2(ma)
 \le4n e^{-2ma}\le4n e^{-a(g-2)}.
\end{align*}
\end{proof}

The proof of Theorem \ref{thm:girth-potts} now follows immediately.

\begin{proof}[Proof of Theorem~\ref{thm:girth-potts}]
Set $\eta=1/100$, which satisfies $0<\eta<1-\rho_\Delta$ since $\rho_\Delta\le\rho_3=2\sqrt2/3$.
The proofs of conditions 1-3 follow exactly as in the proof of Theorem~\ref{thm:potts-application}. Thus, the result follows from Theorem \ref{thm:structural-rank} and the fact that Lemma~\ref{lem:girth-rank} gives
$k_\eta(G)\le4n e^{-a_\Delta(g-2)}$.
\end{proof}

\subsection{Bipartite hard-core model (Proof of Theorems \ref{thm:vt-hardcore} and \ref{thm:girth-hardcore})}
\label{sec:hardcore-application}

\begin{proof}[Proof of Theorem~\ref{thm:vt-hardcore}]
Write $V=V_L\sqcup V_R$ for the bipartition. Relabeling the spin set $[2]$ by $\{0,1\}$, we can encode an independent set $I$ by
\[
    z_v=
    \begin{cases}
        \mathbf{1}\{v\in I\}, & v\in V_L,\\
        1-\mathbf{1}\{v\in I\}, & v\in V_R.
    \end{cases}
\]
On an edge $uv$ with $u\in V_L$ and $v\in V_R$, the permitted
pairs are $00,01,11$. The only permitted
disagreeing pair is $01$, which corresponds to an edge with
neither endpoint occupied. Therefore, we have that
\[
    c(z)=|E|-\Delta|I|=\Delta(n/2-|I|)
\]
and consequently,
\[
    Z_G(\lambda)
    =\lambda^{n/2}
      \sum_{z\in\mathcal F_G}
      e^{-(\log\lambda/\Delta)c(z)}.
\]
We verify the conditions of Theorem~\ref{thm:structural-rank}. For Condition 1, the constant configurations $0\mathbf{1}_n$
and $1\mathbf{1}_n$ encode the independent sets $V_R$ and $V_L$,
respectively, and hence are feasible.

For Condition 2, every boundary neighbor of a monochromatic
component has the opposite spin. Flipping the component makes
each boundary pair diagonal, while its internal pairs remain
diagonal. Both diagonal pairs are permitted, so the recolored
configuration is feasible.

For Condition 3, take
$\mathcal T=\{\mathrm{occ},\mathrm{vac}\}$ and define
\[
    t(v,i)=
    \begin{cases}
        \mathrm{occ}, & (v\in V_L,\ i=1)
                       \text{ or }(v\in V_R,\ i=0),\\
        \mathrm{vac}, & \text{otherwise}.
    \end{cases}
\]
To verify \ref{cond:3a}, suppose $t(v,i)=t(w,j)$.
Vertex transitivity gives an automorphism $\pi$ with $\pi(v)=w$.
Since $G$ is connected and bipartite, $\pi$ either preserves
both parts or exchanges them. Choose the identity spin permutation
in the first case and the flip $\sigma(i)=1-i$ in the second.
The equality of types ensures $\sigma(i)=j$. The action
$(Tz)_{\pi(u)}=\sigma(z_u)$ transports the encoded independent
set by $\pi$, so it preserves every edge constraint.
Thus the exact version of Condition \ref{cond:3a} holds.

For \ref{cond:3b}, choose $z_0=0\mathbf{1}_n$. Both types occur
$n/2$ times in $z_0$. Order each edge $uv$ with $u\in V_L$,
$v\in V_R$, and define
\[
    a_{uv,\mathrm{occ}}(i,j)=\frac{i-j}{\Delta},
    \qquad
    a_{uv,\mathrm{vac}}(i,j)=\frac{j-i}{\Delta}.
\]
For every feasible $z$, regularity gives
\begin{align*}
    \sum_{uv\in E}a_{uv,\mathrm{occ}}(z_u,z_v)
    &=\sum_{u\in V_L}z_u-\sum_{v\in V_R}z_v\\
    &=|I|-n/2\\
    &=n_{\mathrm{occ}}(z)-n_{\mathrm{occ}}(z_0).
\end{align*}
Negating this identity gives the corresponding identity for
the vacant type. For \ref{cond:3c}, note that every permitted pair satisfies
\[
    |a_{uv,\mathrm{occ}}(i,j)|
    +|a_{uv,\mathrm{vac}}(i,j)|
    =\frac{2}{\Delta}\mathbf{1}\{i\ne j\}.
\]
Thus the balance bound holds with $B=2/\Delta$. All conditions are therefore satisfied. Taking
\[
    \lambda_0
    =\exp\!\left(\Delta\theta_0(2,\eta,2/\Delta)\right),
\]
Theorem~\ref{thm:structural-rank} gives the stated guarantees for the normalized model. Multiplying its partition function estimate by $\lambda^{n/2}$ preserves relative error, and decoding the sampled configuration preserves total variation distance.
\end{proof}

\begin{proof}[Proof of Theorem~\ref{thm:girth-hardcore}]
Set $\eta=1/100$. Since $\rho_\Delta\le\rho_3=2\sqrt2/3$ for
$\Delta\ge3$, we have $0<\eta<1-\rho_\Delta$, and the exponent $a$ in Lemma~\ref{lem:girth-rank} equals $a_\Delta$. We once again verify the conditions of Theorem~\ref{thm:structural-rank}. Choose a bipartition $V=V_L\sqcup V_R$ and use the encoding from
Theorem~\ref{thm:vt-hardcore}: the permitted pairs from $V_L$ to $V_R$
are $00,01,11$, and $\theta=\log\lambda/\Delta$.
Both constant configurations are feasible, and flipping a monochromatic
component makes its boundary pairs diagonal. This proves Conditions 1
and 2.

Use occupied and vacant root types, exactly as in that application.
Put $R=g/2-1$. Since $2R+1<g$, every induced radius-$R$ neighborhood is the same
rooted $\Delta$-regular tree truncated at depth $R$.
For roots of the same type, choose a rooted tree isomorphism; use
identity spins if their roots lie in the same bipartition class and
flip spins otherwise. Then $\sigma(i)=j$, and all edge constraints
are preserved, including when the bipartition is exchanged. All ambient
vertex degrees equal $\Delta$. This proves the local version of \ref{cond:3a}.

For \ref{cond:3b}--\ref{cond:3c}, choose $z_0=0\mathbf1_n$ and the same charges
\[
 a_{uv,\mathrm{occ}}(i,j)=\frac{i-j}{\Delta},\qquad
 a_{uv,\mathrm{vac}}(i,j)=\frac{j-i}{\Delta}
 \quad(u\in V_L,\ v\in V_R).
\]
Regularity gives their sum as $|I|-n/2$ and their
combined absolute value is $(2/\Delta)\mathbf1\{i\ne j\}$.
Thus $B=2/\Delta$, without any use of transitivity.

The meta theorem gives time exponential in
$k_\eta(G)+n e^{-\Delta(g/2-1)}$. Lemma~\ref{lem:girth-rank}
bounds the rank by $4n e^{-a_\Delta(g-2)}$. Since
\[
 a_\Delta<\operatorname{arcosh}(1/\rho_\Delta)
 =\tfrac12\log(\Delta-1)<\Delta/2,
\]
the local-symmetry error is dominated by the stated girth exponent and so the result follows.
\end{proof}

\section{Preliminaries}\label{sec:preliminaries}

\subsection{Polymer models and cluster expansions}

We begin by defining a polymer model.

\begin{definition}[Polymer Model] \label{def:abstract-polymer-model}
Fix a host graph \(H\). Choose a family
\(\cP=\cP(H)\) of possibly labeled nonempty connected subsets of \(V(H)\) and assign each
\(\gamma\in\cP\) a weight \(w_\gamma\in\mathbb C\). 
Throughout, distances, sizes, vertex membership, and unions of polymers refer to their supports.

We say that two polymers \(\gamma,\gamma'\in\cP\)
are compatible if and only if
$\operatorname{dist}_H(\gamma,\gamma')>1$. We call a set of mutually compatible polymers a configuration and write \(\Omega(\cP)\) to be the set of all configurations.

A polymer model is then defined by its partition function and associated Gibbs distribution:
\begin{equation}
\Xi
:=
\sum_{\Gamma\in\Omega(\cP)}
\prod_{\gamma\in\Gamma}w_\gamma, \qquad \nu(\Gamma)
=
\frac{1}{\Xi}
\prod_{\gamma\in\Gamma}w_\gamma.
\label{eq:polymer-partition-function}
\end{equation}
\end{definition}

We note that in the above definition, the empty configuration belongs to \(\Omega(\cP)\) and contributes
\(1\) to \(\Xi\). A useful fact about polymer models is that they admit a cluster expansion. 

\begin{definition}[Cluster Expansion]\label{def:cluster}
Let
\(\boldsymbol{\gamma}=(\gamma_1,\ldots,\gamma_k)\in\cP^k\) be a tuple of polymers (with repetition allowed). Let \(F_{\boldsymbol{\gamma}}\) be the incompatibility graph formed from \(\boldsymbol{\gamma}\) on vertex set \([k]\), with edges between positions
\(i,j\) if their polymers are incompatible. 

We say that \(\boldsymbol{\gamma}\) is a  \emph{cluster} when its graph $F_{\boldsymbol{\gamma}}$ is connected. We denote
$|\boldsymbol{\gamma}|:=\sum_{i=1}^k|\gamma_i|$ to be its size and associate the coefficient to each cluster:
\begin{equation}\label{eq:ursell-definition}
\phi(\boldsymbol{\gamma})
:=
\frac{1}{k!}
\sum_{\substack{
A\subseteq E(F_{\boldsymbol{\gamma}})\\
([k],A)\ \mathrm{connected}
}}
(-1)^{|A|}
\end{equation}
The cluster expansion is the following identity
\begin{equation}
\log\Xi
=
\sum_{k\geq 1}
\sum_{\boldsymbol{\gamma}\in\cP^k}
\phi(\boldsymbol{\gamma})
\prod_{i=1}^k w_{\gamma_i}.
\label{eq:cluster-expansion}
\end{equation}
\end{definition}

To justify evaluating the cluster expansion, we must verify its convergence. The Koteck\'y--Preiss criterion \cite{KP} gives a sufficient condition for this. The below version of the KP criterion is given by \cite[Theorem~6]{JKP20}.

\begin{theorem}[Koteck\'y--Preiss criterion]
\label{thm:KP}
Consider a polymer model given by \((\cP,w)\) on host graph \(H\). Define the function $\mathfrak g:\cP\to\R_{>0}$ and extend $\mathfrak g$ to tuples by $\mathfrak g(\boldsymbol{\gamma}):=\sum_{i=1}^k\mathfrak g(\gamma_i)$.
If the function $\mathfrak g$ satisfies
\begin{equation}
\sum_{\substack{\gamma'\in\cP\\
\operatorname{dist}_H(\gamma,\gamma')\leq 1}}
|w_{\gamma'}|
\exp\!\left(
|\gamma'|+\mathfrak g(\gamma')
\right)
\leq |\gamma|
\label{eq:KP-criterion}
\end{equation}
for all \(\gamma\in\cP\), the cluster expansion
\eqref{eq:cluster-expansion} is absolutely convergent. Furthermore, every vertex \(v\in V(H)\) satisfies
\begin{equation}
\sum_{\substack{
k\geq1,\ \boldsymbol{\gamma}\in\cP^k\\
F_{\boldsymbol{\gamma}}\ \mathrm{connected}\\
v\in\bigcup_{i=1}^k\gamma_i
}}
|\phi(\boldsymbol{\gamma})|
\prod_{i=1}^k|w_{\gamma_i}|
\exp\!\left(
\mathfrak g(\boldsymbol{\gamma})
\right)
\leq 1.
\label{eq:KP-tail-bound}
\end{equation}
\end{theorem}

To implement an efficient algorithm from a convergent cluster expansion, we must show that we can get a sufficiently good approximation from truncation. This is given by the following corollary.

\begin{corollary}\label{cor:cluster-truncation}
Suppose that the conditions of Theorem~\ref{thm:KP} are met and further that there exists some $\rho>0$ such that the function $\mathfrak g$ satisfies 
$\mathfrak g(\gamma)\ge\rho|\gamma|$ for every polymer.

Then, for every $m\ge0$,
\begin{equation}
\sum_{\substack{
k\geq1,\ \boldsymbol{\gamma}\in\cP^k\\
F_{\boldsymbol{\gamma}}\ \mathrm{connected}\\ |\boldsymbol{\gamma} | \ge m
}}
|\phi(\boldsymbol{\gamma})|
\prod_{i=1}^k|w_{\gamma_i}|
\leq
|V(H)|e^{-\rho m}.
\label{eq:cluster-tail}
\end{equation}
Consequently, truncating the cluster expansion at depth $
m \ge \rho^{-1}\log\!\left(\frac{|V(H)|}{\varepsilon}\right)$
approximates $\log\Xi$ to additive error at most $\varepsilon$. 
For nonnegative real polymer weights and $0<\varepsilon<1/2$, choosing
$m\ge\rho^{-1}\log(2|V(H)|/\varepsilon)$ gives additive error at most
$\varepsilon/2$. Exponentiating the truncated sum then gives an estimate
$\widehat\Xi$ satisfying
\[
    (1-\varepsilon)\Xi\le\widehat\Xi\le(1+\varepsilon)\Xi.
\]
\end{corollary}
\begin{proof}
Since $\mathfrak g(\gamma)\ge \rho|\gamma|$ for every polymer, every cluster
$\boldsymbol{\gamma}=(\gamma_1,\ldots,\gamma_k)$ with
$|\boldsymbol{\gamma}|\ge m$ satisfies
\[
\mathfrak g(\boldsymbol{\gamma})
=\sum_{i=1}^k\mathfrak g(\gamma_i)
\ge \rho\sum_{i=1}^k|\gamma_i|
=\rho|\boldsymbol{\gamma}|
\ge \rho m.
\]
Thus,
\begin{align*}
&\sum_{\substack{
k\ge1,\ \boldsymbol{\gamma}\in\cP^k\\
F_{\boldsymbol{\gamma}}\ \mathrm{connected}\\
|\boldsymbol{\gamma}|\ge m
}}
|\phi(\boldsymbol{\gamma})|
\prod_{i=1}^k|w_{\gamma_i}|
\le
e^{-\rho m}
\sum_{\substack{
k\ge1,\ \boldsymbol{\gamma}\in\cP^k\\
F_{\boldsymbol{\gamma}}\ \mathrm{connected}\\
|\boldsymbol{\gamma}|\ge m
}}
|\phi(\boldsymbol{\gamma})|
\prod_{i=1}^k|w_{\gamma_i}|
e^{\mathfrak g(\boldsymbol{\gamma})}\\
&\le
e^{-\rho m}
\sum_{v\in V(H)}
\sum_{\substack{
k\ge1,\ \boldsymbol{\gamma}\in\cP^k\\
F_{\boldsymbol{\gamma}}\ \mathrm{connected}\\
v\in\bigcup_{i=1}^k\gamma_i
}}
|\phi(\boldsymbol{\gamma})|
\prod_{i=1}^k|w_{\gamma_i}|
e^{\mathfrak g(\boldsymbol{\gamma})}\le |V(H)|e^{-\rho m},
\end{align*}
where the final inequality follows from~\eqref{eq:KP-tail-bound}. Plugging in $m = \rho^{-1}\log\!\left(\frac{|V(H)|}{\varepsilon}\right)$ gives the additive-error bound. For nonnegative real weights, $\Xi\ge1$ because the empty configuration contributes one. Applying the bound with accuracy $\varepsilon/2$ and exponentiating gives the relative-error claim, since $e^{-\varepsilon/2}\ge1-\varepsilon$ and $e^{\varepsilon/2}\le1+\varepsilon$ for $0<\varepsilon<1/2$.
\end{proof}

The above corollary gives us an algorithmic way to utilize the cluster expansion. Recall that our algorithm will leverage multiple cluster expansions over a family of reference cuts. The following lemma bounds the contribution of clusters meeting a specified vertex, which we will use to compare the cluster expansions in our analysis.

For notation, given a polymer $\gamma\in\cP$ and a tuple
$\boldsymbol{\gamma}=(\gamma_1,\ldots,\gamma_k)$, write
$m_\gamma(\boldsymbol{\gamma})=|\{i\in[k]:\gamma_i=\gamma\}|$
for the number of occurrences of $\gamma$ in the tuple.

\begin{lemma}\label{lem:marked-KP}
Consider a polymer model $(\cP,w)$ on host graph $H$, with
$w_\gamma\ge0$ for every $\gamma\in\cP$. Let
$\mathfrak f,\mathfrak g:\cP\to\R_{\ge0}$, and extend
$\mathfrak g$ to tuples by
$\mathfrak g(\boldsymbol{\gamma})=\sum_{i=1}^k\mathfrak g(\gamma_i)$.
Suppose that
\[
\sum_{\substack{\gamma'\in\cP\\
\operatorname{dist}_H(\gamma,\gamma')\le1}}
w_{\gamma'}
\exp\!\left(\mathfrak f(\gamma')+\mathfrak g(\gamma')\right)
\le \mathfrak f(\gamma)
\]
for every $\gamma\in\cP$. Then, every $\gamma\in\cP$ satisfies
\[
\sum_{\substack{
k\ge1,\ \boldsymbol{\gamma}\in\cP^k\\
F_{\boldsymbol{\gamma}}\ \mathrm{connected}
}}
|\phi(\boldsymbol{\gamma})|
m_\gamma(\boldsymbol{\gamma})
\prod_{i=1}^k w_{\gamma_i}
\exp\!\left(\mathfrak g(\boldsymbol{\gamma})\right)
\le w_\gamma
\exp\!\left(\mathfrak f(\gamma)+\mathfrak g(\gamma)\right).
\]
Additionally, for every vertex $v\in V(H)$, it follows that
\[
\sum_{\substack{
k\ge1,\ \boldsymbol{\gamma}\in\cP^k\\
F_{\boldsymbol{\gamma}}\ \mathrm{connected}\\
v\in\bigcup_{i=1}^k\gamma_i
}}
|\phi(\boldsymbol{\gamma})|
\prod_{i=1}^k w_{\gamma_i}
\exp\!\left(\mathfrak g(\boldsymbol{\gamma})\right)
\le
\sum_{\substack{\gamma\in\cP\\v\in\gamma}}
w_\gamma\exp\!\left(\mathfrak f(\gamma)+\mathfrak g(\gamma)\right).
\]
The same bounds hold for labeled polymers, with distances and vertex
membership interpreted through their supports.
\end{lemma}

We note that the above lemma is a specialization of the standard tree sum argument with the Koteck\'y--Preiss condition (see ~\cite[Section~2]{FP07}).
We include a proof for completeness.

\begin{proof}
Set $\widetilde w_\gamma=w_\gamma e^{\mathfrak g(\gamma)}$,
and let $\tau(F)$ denote the number of spanning trees of $F$.
Then, we have that
\begin{align} \label{eq:cp_bound}
&\sum_{k\ge1}
\sum_{\substack{\boldsymbol\gamma\in\cP^k\\
F_{\boldsymbol\gamma}\text{ connected}}}
m_\gamma(\boldsymbol\gamma)
|\phi(\boldsymbol\gamma)|
\prod_{i=1}^k w_{\gamma_i}\,
e^{\mathfrak g(\boldsymbol\gamma)} \nonumber
\\
&=\sum_{k\ge1}\sum_{j=1}^k
\sum_{\substack{\boldsymbol\gamma\in\cP^k\\
F_{\boldsymbol\gamma}\text{ connected}\\
\gamma_j=\gamma}}
|\phi(\boldsymbol\gamma)|
\prod_{i=1}^k\widetilde w_{\gamma_i} \nonumber
\\
&=\sum_{k\ge1}k
\sum_{\substack{\boldsymbol\gamma\in\cP^k\\
F_{\boldsymbol\gamma}\text{ connected}\\
\gamma_1=\gamma}}
|\phi(\boldsymbol\gamma)|
\prod_{i=1}^k\widetilde w_{\gamma_i}
\end{align}
where the first equality follows by expanding
$m_\gamma(\boldsymbol\gamma)=\sum_{j=1}^k\mathbf1\{\gamma_j=\gamma\}$
and using the definition of $\widetilde w$, and the second follows
from invariance under permutations of the tuple positions.

We now apply a standard tree-graph inequality~\cite[Eq.~(4.5)]{FP07}. This states that
\[
    |\phi(\boldsymbol\gamma)|
    \le \frac{\tau(F_{\boldsymbol\gamma})}{k!}.
\]
Plugging this bound into \eqref{eq:cp_bound} gives that

\begin{align} \label{eq:plugged_in}
\sum_{k\ge1}k
\sum_{\substack{\boldsymbol\gamma\in\cP^k\\
F_{\boldsymbol\gamma}\text{ connected}\\
\gamma_1=\gamma}}
|\phi(\boldsymbol\gamma)|
\prod_{i=1}^k\widetilde w_{\gamma_i} \le\sum_{k\ge1}\frac{1}{(k-1)!}
\sum_{\substack{\boldsymbol\gamma\in\cP^k\\
F_{\boldsymbol\gamma}\text{ connected}\\
\gamma_1=\gamma}}
\tau(F_{\boldsymbol\gamma})
\prod_{i=1}^k\widetilde w_{\gamma_i}
\end{align}

Now, let $T_\gamma^{(N)}$ denote the right-hand side above with
the sum restricted to $1\le k\le N$. We prove by induction
on $N$ that
\[
    T_\gamma^{(N)}
    \le \widetilde w_\gamma e^{\mathfrak f(\gamma)}
\]
for every $\gamma\in\cP$.
The case $N=1$ is trivial as $T_\gamma^{(1)}=\widetilde w_\gamma$ and $\mathfrak f\ge0$. For $N\ge2$, removing vertex $1$ from each spanning tree
leaves an unordered collection of rooted subtrees. Each subtree has at most $N-1$ vertices, and its root carries a polymer incompatible with $\gamma$. Dropping the restriction that these subtrees together have at most $N-1$ vertices gives

\begin{align*}
T_\gamma^{(N)}
&\le \widetilde w_\gamma
\sum_{d\ge0}\frac{1}{d!}
\left(
    \sum_{\substack{\gamma'\in\cP\\
    \operatorname{dist}_H(\gamma,\gamma')\le1}}
    T_{\gamma'}^{(N-1)}
\right)^d = \widetilde w_\gamma
\exp\left(
    \sum_{\substack{\gamma'\in\cP\\
    \operatorname{dist}_H(\gamma,\gamma')\le1}}
    T_{\gamma'}^{(N-1)}
\right)\\ 
&\le \widetilde w_\gamma
\exp\left(
    \sum_{\substack{\gamma'\in\cP\\
    \operatorname{dist}_H(\gamma,\gamma')\le1}}
    \widetilde w_{\gamma'}e^{\mathfrak f(\gamma')}
\right) \le \widetilde w_\gamma e^{\mathfrak f(\gamma)}
\end{align*}

where the first inequality follows by distributing the vertex
labels among the subtrees and dropping the restriction on their
combined size, with $1/d!$ correcting for their ordering.
The equality follows from the Taylor expansion of an exponential function, and the second and third inequalities follow from the induction hypothesis and the assumption of the lemma, respectively.

Finally, since all summands are nonnegative, letting $N\to\infty$
bounds the full tree sum by
$\widetilde w_\gamma e^{\mathfrak f(\gamma)}
=w_\gamma e^{\mathfrak f(\gamma)+\mathfrak g(\gamma)}$, proving the first claim.

For the second claim, fix $v\in V(H)$ and observe that every
tuple $\boldsymbol\gamma=(\gamma_1,\ldots,\gamma_k)$ satisfies
\[
    \mathbf 1\left\{v\in\bigcup_{i=1}^k\gamma_i\right\}
    \le
    \sum_{\substack{\gamma\in\cP\\v\in\gamma}}
    m_\gamma(\boldsymbol\gamma).
\]
Multiply by
$|\phi(\boldsymbol\gamma)|\prod_{i=1}^k\widetilde w_{\gamma_i}$
and sum over all tuples with $F_{\boldsymbol\gamma}$ connected.
Applying the first claim to each polymer containing $v$
then gives the desired bound.

\end{proof}

\subsection{Active regression}
\label{subsec:subspace-enumeration}

In active regression, we know a full-column-rank matrix
$X\in\R^{m\times r}$ and may query entries of an unknown response
$y\in\R^m$. The goal is to use few queries to construct a prediction
$\widehat f\in\operatorname{col}(X)$ whose squared error is as close as possible to the optimal least-squares error
\[
 \min_{\beta\in\R^r}\norm{X\beta-y}_2^2
 =\norm{(\Id-Q)y}_2^2,
\]
where $Q$ is the orthogonal projector onto $\operatorname{col}(X)$.

Chen and Price \cite{CP} show that $O(r/\varepsilon)$ queries suffice
for an expected squared error at most $(1+\varepsilon)$ times this
optimum. They select and weight coordinates using only $X$ and
$\varepsilon$, then solve weighted least squares on the observed
responses. The following theorem states their guarantee in terms of
distance from the optimal prediction $Qy$.

\begin{theorem}[Chen--Price {\cite[Theorem~1]{CP}}]
\label{thm:CP}
Let $X\in\R^{m\times r}$ have full column rank, and let
$Q$ be the orthogonal projector onto $\operatorname{col}(X)$.
For every $\varepsilon\in(0,1)$, there is an efficient randomized
procedure depending only on $X$ and $\varepsilon$ that selects
a set $I\subseteq[m]$ of size $O(r/\varepsilon)$ and positive
weights $(w_i)_{i\in I}$.

These indices and weights are such that, for every fixed
response $y\in\R^m$, the weighted least-squares prediction
\[
    \widehat\beta_y\in
    \operatorname*{argmin}_{\beta\in\R^r}
    \sum_{i\in I}w_i(X_{i,*}\beta-y_i)^2,
    \qquad
    \widehat f_y=X\widehat\beta_y
\]
satisfies
\[
    \E\|\widehat f_y-Qy\|_2^2
    \le \varepsilon\|(\Id-Q)y\|_2^2,
\]
where the expectation is over the choice of $I$ and its weights.
\end{theorem}

The next lemma uses a common weighted sample to control the total
regression error across several response vectors.

\begin{lemma}\label{lem:capped-regression}
Let $U\in\R^{n\times r}$ have orthonormal columns,
with $r\ge1$. There is an efficient randomized procedure depending only on $U$
that selects a set $I\subseteq[n]$ of at most $Cr$ coordinates
and positive weights $(w_i)_{i\in I}$, where $C$ is an absolute
constant.

For every fixed $Y\in\R^{n\times q}$, with probability
at least $0.99$ the procedure returns coordinates and weights
such that fitting each column of $Y$ by weighted least squares
using this common sample yields a prediction $\widehat Y$ satisfying
\[
    \|\widehat Y-UU^{\mathsf T}Y\|_F^2
    \le \frac14\|(\Id-UU^{\mathsf T})Y\|_F^2.
\]
\end{lemma}

\begin{proof}
Apply Theorem~\ref{thm:CP} with $X=U$ and
$\varepsilon=1/400$. This selects and weights at most $Cr$
coordinates depending only on $U$, for an absolute constant $C$.
Use these same coordinates and weights to fit every column
of $Y$.

For each fixed $Y\in\R^{n\times q}$, linearity of
expectation gives
\begin{align*}
    \E\|\widehat Y-UU^{\mathsf T}Y\|_F^2
    &= \sum_{j=1}^q
       \E\|\widehat Y^j-UU^{\mathsf T}Y^j\|_2^2\\
    &\le \frac1{400}
       \sum_{j=1}^q\|(\Id-UU^{\mathsf T})Y^j\|_2^2\\
    &= \frac1{400}\|(\Id-UU^{\mathsf T})Y\|_F^2.
\end{align*}
The result follows by Markov's inequality. 
\end{proof}

\section{Proof of Theorem \ref{thm:structural-rank}}\label{sec:main-proof}

Recall the spectral subspace $S$ and the notation $r=\dim S$ and $\Pi_S$ for its orthogonal projector.

\subsection{\texorpdfstring{$C^3$}{C3} Part 1: Constructing \texorpdfstring{$\mathcal Q$}{Q} via active regression}\label{sec:c3po-spectral}

Fix a cut budget $b$ and recall the reference family $\mathcal H_b$ from \eqref{eq:overview-reference-family}. We first construct an intermediate family
$\mathcal Q\subseteq[q]^V$ of cuts close to $\mathcal{H}_b$. That is, for every $h\in\mathcal H_b$,
there is a candidate $x \in \mathcal{Q}$ such that
\[
 \vol(\{v:x_v\ne h_v\})\le \frac{16b}{\eta}.
\]

The following algorithm describes how to use subspace enumeration (combined with active regression) to achieve this.

\begin{algorithm}[H]
\caption{Constructing candidates by active regression}
\label{alg:approximate-references}
\begin{algorithmic}[1]
\Require An orthonormal basis $U$ of $S$, vertex degrees $d_v=\deg_G(v)$, and $0<\zeta<1$.
\Ensure A candidate family $\cQ$ as in Lemma~\ref{lem:approximate-references}.
\State Set $\cQ\gets\varnothing$ and
$T\gets\lceil n\log q+\log(2/\zeta)\rceil$.
\For{$T$ independent repetitions}
\State Use the procedure in Lemma~\ref{lem:capped-regression} to select
coordinates $I$ and weights $(w_v)_{v\in I}$.
\For{each assignment $\sigma\in[q]^I$}
\For{each color $j\in[q]$}
\State Set $y_{v,j}\gets\sqrt{d_v}\one\{\sigma_v=j\}$ for $v\in I$.
\State Compute $\widehat\beta_j\in\argmin_{\beta\in\R^r}
\sum_{v\in I}w_v(U_{v,*}\beta-y_{v,j})^2$ and
$\widehat Y^j\gets U\widehat\beta_j$.
\EndFor
\State Add $x$ to $\cQ$, where $x_v\in\operatorname*{arg\,max}_{j\in[q]}(\widehat Y^j)_v$
for every $v$, with fixed tie-breaking.
\EndFor
\EndFor
\State \Return $\cQ$.
\end{algorithmic}
\end{algorithm}

\begin{lemma}\label{lem:approximate-references}
For every $0<\zeta<1$, Algorithm~\ref{alg:approximate-references}
runs in time
$q^{O(r)}\operatorname{poly}(n,\log(1/\zeta))$
and constructs a family $\cQ\subseteq[q]^V$ of at most this size.
With probability at least $1-\zeta/2$, every configuration
$h\in[q]^V$ has a candidate $x\in\cQ$ satisfying
\[
    \vol(\{v:x_v\ne h_v\})\le\frac{16c(h)}{\eta}
\]
\end{lemma}

\begin{proof}
Fix a configuration $h\in[q]^V$ and let $Y_h=[Y_h^1\ \cdots\ Y_h^q]\in\R^{n\times q}$, where $(Y_h^a)_v=\sqrt{d_v}\one\{h_v=a\}$ for $v\in V$ and $a\in[q]$. Note that by construction,
\[
    \sum_{a=1}^q (Y_h^a)^{\mathsf T} L_G Y_h^a=2c(h),
\]
since each bichromatic edge contributes twice to the sum.
Because every eigenvalue of $L_G$ on $S^\perp$ is at least $\eta/2$, we have that
\begin{equation}\label{potts:eqtag:6}
    \mathcal E:=\|(\Id-\Pi_S)Y_h\|_F^2
    \le \frac2\eta\sum_{a=1}^q (Y_h^a)^{\mathsf T} L_G Y_h^a
    =\frac{4c(h)}{\eta}.
\end{equation}

In each repetition, algorithm~\ref{alg:approximate-references} considers the assignment $\sigma=h|_I$ on the sampled coordinates $I$. For this assignment,
Lemma~\ref{lem:capped-regression} gives, with probability at least
$0.99$, a prediction $\widehat Y$ satisfying
\[
    \|\widehat Y-\Pi_S Y_h\|_F^2\le\mathcal E/4.
\]
Every column of $\widehat Y$ lies in $S$. Thus, on this event,
orthogonality gives
\[
    \|\widehat Y-Y_h\|_F^2
    =\|\widehat Y-\Pi_S Y_h\|_F^2
      +\|(\Id-\Pi_S)Y_h\|_F^2
    \le 2\mathcal E.
\]
Now, let $x$ be the candidate obtained by taking a largest entry
of each row of $\widehat Y$. If $x_v\ne h_v$, then
$\widehat Y_{v,x_v}\ge\widehat Y_{v,h_v}$, and hence
\[
    \|\widehat Y_{v,*}-(Y_h)_{v,*}\|_2^2
    \ge (\widehat Y_{v,h_v}-\sqrt{d_v})^2
        +\widehat Y_{v,x_v}^{\,2}
    \ge \frac{d_v}{2}.
\]
Summing over misclassified vertices yields
\begin{equation}\label{potts:eqtag:7}
    \vol(\{v:x_v\ne h_v\})
    \le 2\|\widehat Y-Y_h\|_F^2
    \le 4\mathcal E
    \le \frac{16c(h)}{\eta}.
\end{equation}

A fixed configuration therefore has no suitable candidate
after $T$ independent repetitions with probability at most
$(0.01)^T$. By the choice of $T$ and a union bound over
the $q^n$ configurations,
\[
    \Pr\bigl(\text{some configuration has no suitable candidate}\bigr)
    \le q^n(0.01)^T
    \le q^n e^{-T}
    \le \frac{\zeta}{2}.
\]

Each repetition considers at most $q^{Cr}$ assignments, each taking polynomial time. Thus the family size and running time are bounded by
$q^{O(r)}\operatorname{poly}(n,\log(1/\zeta))$.
\end{proof}

\subsection{\texorpdfstring{$C^3$}{C3} Part 2: Recovering \texorpdfstring{$\mathcal H_b$}{H} from \texorpdfstring{$\mathcal Q$}{Q}}\label{sec:reference-construction}

We now describe how to recover $\mathcal H_b$ from $\mathcal Q$.
By Lemma \ref{lem:approximate-references}, we have that with high probability, every reference in $\mathcal H_b$ differs
from some candidate in $\mathcal Q$ on a set of volume at most
$16b/\eta$. We shall now utilize stability to enumerate the changes needed
to recover these references.

We first prove Lemma~\ref{lem:overview-reference-properties}: the reference $\kappa(z)$ from \eqref{eq:overview-reference-map} is stable, and every connected component on which it differs from $z$ defines a polymer.

\begin{proof}[Proof of Lemma~\ref{lem:overview-reference-properties}]
Let $h=\kappa(z)$, and let $T$ be a connected component of
$\{v:z_v\ne h_v\}$. Write $h^{T,z}$ for the configuration
that agrees with $z$ on $T$ and with $h$ outside $T$.

Observe that $h^{T,z}$ is feasible. Clearly every edge with
both endpoints in $T$ or both endpoints outside $T$ is satisfied. For a boundary edge $uv$ with
$u\in T$ and $v\notin T$, we have $h_v=z_v$. Otherwise,
$v$ would belong to the same connected component of the
disagreement set as $u$, contradicting $v\notin T$.
Thus the spin pair on $uv$ in $h^{T,z}$ also agrees with
that in $z$. Since both $h$ and $z$ are feasible,
$h^{T,z}$ is feasible.

Recoloring $T$ removes exactly its vertices from the
disagreement set with $z$. Therefore,
\[
    \vol(\{v:h^{T,z}_v\ne z_v\})
    =
    \vol(\{v:h_v\ne z_v\})-\vol(T).
\]
Since $h=\kappa(z)$ minimizes the defining objective, we have
\[
    c(h)+\alpha\vol(\{v:h_v\ne z_v\})
    \le
    c(h^{T,z})
    +\alpha\bigl[\vol(\{v:h_v\ne z_v\})-\vol(T)\bigr].
\]
Rearranging gives $c(h^{T,z})-c(h)\ge\alpha\vol(T)$ and thus $(T,z|_T)$ is a polymer around $h$. The connected components of the disagreement set are disjoint and nonadjacent, so these polymers are compatible.

We next prove that $h$ is stable. Let $C$ be a monochromatic
component of $h$. If $\partial C=\varnothing$, stability
is immediate. Otherwise, at most $q-1$ spins
occur across its boundary, so some spin occurs on at least
$|\partial C|/(q-1)$ boundary edges. Let $h'$ be obtained
by recoloring all of $C$ with this spin.
Condition~\ref{cond:2} ensures that $h'$ is feasible.

Because every boundary edge of $C$ was a cut edge before this
recoloring, we have that at least $|\partial C|/(q-1)$ of these edges
are no longer cut edges. Thus,
\[
    c(h')-c(h)\le-\frac{|\partial C|}{q-1}.
\]
Moreover, as only vertices in $C$ can change their agreement
with $z$, the volume penalty increases by at most
$\alpha\vol(C)$. Optimality of $h$ therefore gives
\begin{align*}
    0
    &\le c(h')-c(h)
    +\alpha\bigl[
        \vol(\{v:h'_v\ne z_v\})
        -\vol(\{v:h_v\ne z_v\})
    \bigr]\\
    &\le -\frac{|\partial C|}{q-1}+\alpha\vol(C).
\end{align*}
Consequently, $|\partial C| \le\alpha(q-1)\vol(C) =\frac{\eta}{16}\vol(C)$ as desired.
\end{proof}

Although Lemma~\ref{lem:overview-reference-properties} guarantees a stable reference for every feasible configuration, it does not show how to find these references efficiently. In particular, though every reference in $\mathcal H_b$ differs from some candidate in $\mathcal Q$ on a set of volume at most $16b/\eta$, directly searching over these changes can incur a factor of $n^{O(b/\eta)}$. We instead work with the connected components of the disagreement set. The following lemma shows that, once a vertex in a component is known, we can enumerate the possible components containing that vertex in time exponential in the volume budget and polynomial in $n$.

\begin{lemma}\label{volume:lem:encode} Fix a vertex $v\in V$ and a budget $L\ge0$. There are at most $2^{L+1}$ connected sets $T\subseteq V$ with $v\in T$ and $\vol(T)\le L$, and all such sets can be enumerated in time $2^{L+O(1)}\operatorname{poly}(n)$.
\end{lemma}

The following depth-first search branches over included and excluded vertices until the volume budget is exceeded. A state $(T,X)$ records the vertices included in $T$ and those excluded in $X$.

\begin{algorithm}[H]
\caption{Enumerating connected sets within a volume budget}
\label{alg:connected-sets}
\begin{algorithmic}[1]
\Require A graph $G$, a vertex $v\in V$, and a budget $L\ge0$.
\Ensure All connected sets containing $v$ with volume at most $L$.
\State Fix an ordering of $V$.
\State Set $\mathcal L\gets\varnothing$.
\State Initialize a stack $\mathcal W$ with the state
$(\{v\},\varnothing)$.
\While{$\mathcal W$ is nonempty}
    \State Pop a state $(T,X)$ from $\mathcal W$.
    \If{$\vol(T)\le L$}
        \State Let $F$ be the vertices outside $T\cup X$
        having a neighbor in $T$.
        \If{$F=\varnothing$}
            \State Add $T$ to $\mathcal L$.
        \Else
            \State Let $u$ be the first vertex of $F$
            in the fixed ordering.
            \State Push $(T\cup\{u\},X)$ onto $\mathcal W$.
            \State Push $(T,X\cup\{u\})$ onto $\mathcal W$.
        \EndIf
    \EndIf
\EndWhile
\State \Return $\mathcal L$.
\end{algorithmic}
\end{algorithm}

\begin{proof}[Proof of Lemma~\ref{volume:lem:encode}]
Correctness of the included sets from Algorithm \ref{alg:connected-sets} follows immediately from the fact that every output is connected, contains $v$, and has volume at most $L$ by the definition of the algorithm. To see that every such set $T_*$ is included, follow the branch that includes a queried vertex exactly when it belongs to $T_*$. This branch is never pruned, and so connectedness ensures that it eventually outputs $T_*$.

To bound the runtime and search tree, consider a node with included set $T$. Each preceding decision queried a distinct vertex adjacent to an included vertex. Since included vertices are never removed, these decisions can be charged to distinct adjacency-list entries of vertices in $T$. Thus every node passing the budget check has depth at most $\vol(T)\le\lfloor L\rfloor$. Thus, the number of output sets is at most $\sum_{j=0}^{\lfloor L\rfloor}2^j<2^{L+1}.$ As pruned nodes add at most one further level, the entire search tree has $2^{L+O(1)}$ nodes. As each node takes polynomial time to process, the running-time bound follows immediately.
\end{proof}

The preceding lemma enumerates a connected component once a root is specified. Choosing roots independently for $t$ components, however, introduces a factor of $n^t$. However, we can avoid this factor by sampling vertices from which to start Algorithm \ref{alg:connected-sets}. We now describe this sampling procedure.

\begin{algorithm}[H]
\caption{Sampling roots for stable references}
\label{alg:component-vertices}
\begin{algorithmic}[1]
\Require The rank-$r$ projector $\Pi_S$, the number of spins
$q$, and $0<\zeta<1$.
\Ensure A list $\mathcal A$ of ordered $r$-tuples of vertices.
\State Set $\omega_v\gets(\Pi_S)_{vv}/r$ for each $v\in V$.
\State Set
$T\gets\lceil(2e)^r(n\log q+\log(2/\zeta)+1)\rceil$.
\State Set $\mathcal A\gets\varnothing$.
\For{$t=1,\ldots,T$}
    \State Sample $v_1,\ldots,v_r$ independently
    from the distribution $(\omega_v)_{v\in V}$.
    \State Append $(v_1,\ldots,v_r)$ to $\mathcal A$.
\EndFor
\State \Return $\mathcal A$.
\end{algorithmic}
\end{algorithm}

\begin{lemma}\label{lem:component-vertices}
Every stable reference $h$ has at most $r$ monochromatic components. Furthermore, for every $0<\zeta<1$, Algorithm~\ref{alg:component-vertices} runs in time
\[
    \exp(O(r))\operatorname{poly}(n,\log(1/\zeta))
\]
and returns a list $\mathcal A$ of at most this size. With probability at least $1-\zeta/2$, for every stable reference $h$, the list $\mathcal A$ contains a tuple with at least one vertex in each monochromatic component of $h$.
\end{lemma}

\begin{proof}[Proof of Lemma~\ref{lem:component-vertices}]
We first show that every $h$ has at most $r$ monochromatic components. Fix a stable reference $h$, and let $C_1,\ldots,C_s$
be its monochromatic components. For any coefficients
$a_1,\ldots,a_s$, set
\[
    f=D_G^{1/2}\sum_{i=1}^s a_i\one_{C_i}.
\]
Edges within a component contribute zero to $f^{\mathsf T}L_Gf$.
On edges between components, using
$(a_i-a_j)^2\le2a_i^2+2a_j^2$ gives
\[
    f^{\mathsf T}L_Gf
    \le 2\sum_{i=1}^s a_i^2|\partial C_i|
    \le \frac{\eta}{8}\sum_{i=1}^s a_i^2\vol(C_i)
    = \frac{\eta}{8}\|f\|_2^2,
\]
where the second inequality uses the stability bound~\eqref{eq:overview-stability}.
Furthermore, since the components have disjoint supports, the above vectors form an $s$-dimensional space. If $s>r$, this space
would contain a nonzero vector orthogonal to $S$. Such a vector satisfies $f^{\mathsf T}L_Gf\ge(\eta/2)\|f\|_2^2$, a contradiction. Thus $s\le r$.

We now analyze the sampling procedure. For a monochromatic
component $C$, set
\[
    f_C=\frac{D_G^{1/2}\one_C}{\sqrt{\vol(C)}}.
\]
Since $\sum_{v\in C}d_v=\vol(C)$, we have $\|f_C\|_2=1$.
Thus,
\begin{align}
    \sum_{v\in C}(\Pi_S)_{vv}
    \ge f_C^{\mathsf T}\Pi_S f_C \ge 1-\frac{2}{\eta}f_C^{\mathsf T}L_Gf_C =1-\frac{2|\partial C|}{\eta\vol(C)}
    \ge\frac78
    \label{potts:eqtag:9}
\end{align}
where the first inequality follows from Cauchy--Schwarz (using that $f_C$ is supported on $C$), the second uses the eigenvalue bound on $S^\perp$, and the last uses the stability bound~\eqref{eq:overview-stability}.

Note that the probabilities $\omega_v=(\Pi_S)_{vv}/r$ in Algorithm \ref{alg:component-vertices} form a valid distribution as they sum to $1$ since $\operatorname{tr}(\Pi_S)=r$. Additionally, note that by \eqref{potts:eqtag:9}, a single draw lands in any fixed component with probability at least $1/(2r)$. Thus, for a fixed stable reference with $s\le r$ components,
the first $s$ entries of a tuple contain one vertex from each component with probability at least
\[
    s!\left(\frac1{2r}\right)^s
    \ge \frac{r!}{(2r)^r}
    \ge (2e)^{-r}
\]
as there are $s!$ possible orders of the components. The second inequality holds because $j!/(2r)^j$ decreases for $j\le r$, and the last follows from the fact $r!\ge(r/e)^r$.

The algorithm samples $T$ independent tuples. Thus, by taking a union bound over the at most $q^n$ stable references, the probability that some reference has no tuple containing
a vertex from each of its components is at most
\[
    q^n\bigl(1-(2e)^{-r}\bigr)^T
    \le \exp\bigl(n\log q-T(2e)^{-r}\bigr)
    \le \zeta/2,
\]
by the choice of $T$ as desired.
\end{proof}

Having established Lemmas~\ref{lem:overview-reference-properties}, \ref{lem:approximate-references}, \ref{volume:lem:encode}, and~\ref{lem:component-vertices}, we now have the ingredients needed to recover $\mathcal H_b$ from $\mathcal Q$. In particular, Lemma~\ref{lem:overview-reference-properties} shows that every feasible configuration can be represented by compatible polymers around some stable reference and Lemma~\ref{lem:component-vertices} supplies a list of $r$-tuples containing, for each stable reference, a tuple with a vertex in each monochromatic component. By repeatedly applying Lemma~\ref{volume:lem:encode} at these vertices and the endpoints of the candidate's cut edges, we can enumerate the connected corrections within the volume budget and keep the resulting configurations that are feasible, stable, and have cut at most $b$. Formally, this procedure is as follows.

\begin{algorithm}[H]
\caption{Enumerating stable references with bounded cut}
\label{alg:stable-references}
\begin{algorithmic}[1]
\Require $G$, the model's feasibility constraints, the projector $\Pi_S$,
an integer cut budget $b\ge0$, $0<\zeta<1$, and the family $\mathcal Q$
returned by Algorithm~\ref{alg:approximate-references} with parameter $\zeta$.
\Ensure A set $\mathcal S\subseteq\mathcal H_b$, equal to
$\mathcal H_b$ with probability at least $1-\zeta$.
\State Set $\mathcal A$ to the output of
Algorithm~\ref{alg:component-vertices} on $\Pi_S$, $q$, and $\zeta$.
\State Set $L_{\rm corr}\gets\lceil16b/\eta\rceil$ and $\mathcal S\gets\varnothing$.
\For{each $x\in\mathcal Q$ with $c(x)\le b+L_{\rm corr}$ and each tuple $A\in\mathcal A$}
    \State Let $\mathcal R$ consist of the distinct vertices in $A$ and the
    endpoints of the cut edges of $x$.
    \For{each $(w_v)_{v\in \mathcal R}\in\mathbb Z_{\ge0}^{\mathcal R}$ with $\sum_{v\in \mathcal R}w_v\le L_{\rm corr}$}
        \State For each $v\in \mathcal R$, set $\mathcal L_v\gets\{\varnothing\}$ if $w_v=0$;
        otherwise, set $\mathcal L_v$ to the output of
        Algorithm~\ref{alg:connected-sets} on $G$, $v$, and $w_v$.
        \For{each choice $K_v\in\mathcal L_v$ whose nonempty sets are pairwise
        disjoint and nonadjacent}
            \State Set $K\gets\bigcup_{v\in \mathcal R}K_v$.
            \For{each $\sigma\in[q]^K$ with $\sigma_u\ne x_u$ for all $u\in K$}
                \State Let $z$ agree with $\sigma$ on $K$ and with $x$ outside $K$.
                \State Add $z$ to $\mathcal S$ if it is feasible, stable, and $c(z)\le b$.
            \EndFor
        \EndFor
    \EndFor
\EndFor
\State \Return $\mathcal S$, removing duplicates.
\end{algorithmic}
\end{algorithm}

\begin{lemma}\label{potts:lem:enumeration}
For every integer $b\ge0$ and $0<\zeta<1$, Algorithms~\ref{alg:approximate-references} and~\ref{alg:stable-references}
output $\mathcal S\subseteq\mathcal H_b$ equal to $\mathcal H_b$
with probability at least $1-\zeta$, in time and list size
\[
 \exp(O_{q,\eta}(r+b))\operatorname{poly}(n,\log(1/\zeta)).
\]
Denote $A_{\rm vol}=2\log2+\frac{16}{\eta}\log(32(q-1))$.
There are constants $C_0,C_1,C_2$, depending only on $q,\eta$, such
that, with $M=e^{C_0}n^{C_1}e^{C_2r}$,
\begin{equation}\label{eq:overview-reference-count}
 |\mathcal H_b|\le M e^{A_{\rm vol}b}.
\end{equation}
\end{lemma}

\begin{proof}[Proof of Lemma~\ref{potts:lem:enumeration}]
Every output belongs to $\mathcal H_b$ by the final test in
Algorithm~\ref{alg:stable-references}. To prove that every member
of $\mathcal H_b$ is output, assume the constructions in
Lemmas~\ref{lem:approximate-references} and~\ref{lem:component-vertices}
succeed, and set $L_{\rm corr}=\lceil16b/\eta\rceil$.
Fix $h\in\mathcal H_b$. There are a candidate $x\in\mathcal Q$ with
$\vol(\{v:h_v\ne x_v\})\le16b/\eta$ and a tuple $A\in\mathcal A$
containing a vertex in each monochromatic component of $h$.
Changing the spins on a set changes the cut by at most its volume, so
\[
    c(x)\le c(h)+\vol(\{v:h_v\ne x_v\})\le b+L_{\rm corr}.
\]
Thus the algorithm considers this candidate and tuple.

Let $\mathcal R$ be their set of roots. We show that every connected component
$K$ of $\{v:h_v\ne x_v\}$ contains a vertex of $\mathcal R$.
If an edge incident to $K$ is a cut edge of $x$, its endpoint in $K$
is a root. Otherwise, $x$ has a constant spin $i$ on $K$ and all its
neighbors. Every vertex of $K$ has spin different from $i$ in $h$,
whereas every neighbor outside $K$ has spin $i$: such a neighbor
agrees in $h$ and $x$, since otherwise it would belong to $K$.
Consequently, every monochromatic component of $h$ intersecting $K$
is contained in $K$. The tuple $A$ contains a vertex in each such
component, so $K$ contains a root in this case as well.

Fix an ordering of $\mathcal R$ and assign each disagreement component to its
first root. Give that root budget $\vol(K)$ and all unused roots
budget zero. These budgets total at most $L_{\rm corr}$, so the algorithm
considers this allocation. Algorithm~\ref{alg:connected-sets}
enumerates each assigned component by Lemma~\ref{volume:lem:encode}.
The actual components are disjoint and nonadjacent, and the algorithm
tries their replacement spins from $h$. It therefore outputs $h$.
The two preliminary constructions each fail with probability at most
$\zeta/2$, so a union bound gives the claimed success probability.

We next bound the running time and list size. For any candidate and
tuple considered by the algorithm, the number $p=|\mathcal R|$ of roots satisfies
\[
    p\le r+2c(x)\le r+2b+2L_{\rm corr}.
\]
There are $\binom{L_{\rm corr}+p}{p}\le2^{L_{\rm corr}+p}$ allocations of nonnegative integer
budgets totaling at most $L_{\rm corr}$. For a positive budget $w$,
Lemma~\ref{volume:lem:encode} gives at most $2^{w+1}\le4^w$
connected sets at a root; a zero budget gives only the empty set.
Thus each allocation gives at most $4^{L_{\rm corr}}$ choices of sets.
Their total volume is at most $L_{\rm corr}$, so their replacement spins
contribute at most $(q-1)^{L_{\rm corr}}$ choices. The number of refinements
per candidate and tuple is therefore at most
\begin{align*}
    2^{L_{\rm corr}+p}[4(q-1)]^{L_{\rm corr}}
    \le 2^{r+2b}[32(q-1)]^{L_{\rm corr}}\le 32(q-1)\exp\bigl(r\log2+A_{\rm vol}b\bigr),
\end{align*}
where the last inequality uses $L_{\rm corr}\le16b/\eta+1$.
Algorithm~\ref{alg:connected-sets} generates each rooted family in
$2^{w+O(1)}\operatorname{poly}(n)$ time, and each refinement can be
constructed and checked in polynomial time. Multiplying by the sizes
of $\mathcal Q$ and $\mathcal A$, which together are
$\exp(O_q(r))\operatorname{poly}(n,\log(1/\zeta))$, proves the
running-time and list-size bounds.

Finally, fix $\zeta=1/2$. A successful realization exists, and its
output is exactly $\mathcal H_b$. The preceding refinement bound,
together with the bounds on $|\mathcal Q|$ and $|\mathcal A|$,
gives $|\mathcal H_b|\le Me^{A_{\rm vol}b}$ for constants
$C_0,C_1,C_2$ depending only on $q,\eta$.
\end{proof}

\subsection{\texorpdfstring{$P$}{P}: Polymer Models and their Convergence}\label{sec:gas-pressure}
Having constructed the family $\mathcal H_b$, we now study
the cluster expansion around each reference $h\in\mathcal H_b$.
Throughout, we use the polymer model from
Definition~\ref{def:overview-polymer-model}, with partition function $\Xi_h$ and distribution $\nu_h$. We first verify the Koteck\'y--Preiss condition, then use the truncated expansion to approximate $\Xi_h$ and sample from $\nu_h$.

For notation, write $S_\gamma$ for the support of a polymer $\gamma$, and set $\vartheta=\theta\alpha$. Recall that the Peierls-type condition in definition~\ref{def:peierls} gives
\[
    w_h(\gamma)\le e^{-\vartheta\vol(S_\gamma)}.
\] 
For a cluster $\boldsymbol\gamma=(\gamma_1,\ldots,\gamma_k)$,
define its union and total volume by
$U(\boldsymbol\gamma)=\bigcup_{i=1}^k S_{\gamma_i}$ and $  \vol(\boldsymbol\gamma)=\sum_{i=1}^k\vol(S_{\gamma_i})$ where total volume counts multiplicities. Additionally, we will write
\[
    \Phi_h(\boldsymbol\gamma)
    =\phi(\boldsymbol\gamma)\prod_{i=1}^k w_h(\gamma_i)
\]
for its contribution to the cluster expansion for $\log\Xi_h$.

\begin{lemma}\label{lem:boosted-KP}
Suppose $\vartheta\ge2+\log(6(q-1))$.
For every $h \in \mathcal{H}_b$, the polymer model from \ref{def:overview-polymer-model} satisfies the Koteck\'y--Preiss condition with
$\mathfrak f(\gamma)=\mathfrak g(\gamma)=\vol(S_\gamma)$.
Moreover, for every vertex $v$,
\begin{equation}\label{eq:volume-root-bound}
    \sum_{\boldsymbol\gamma:\,v\in U(\boldsymbol\gamma)}
    |\Phi_h(\boldsymbol\gamma)|e^{\vol(\boldsymbol\gamma)}
    \le\frac12,
\end{equation}
where the sum is over clusters.
Consequently, the cluster expansion for $\log\Xi_h$ converges absolutely, and omitting clusters of total volume greater than $L$ incurs error at most $(n/2)e^{-L}$.
\end{lemma}

\begin{proof}
We first verify the Koteck\'y--Preiss condition.
Fix a vertex $v$. The branching proof of Lemma~\ref{volume:lem:encode} gives at most $2^m$ connected supports containing $v$ with volume $m$, each of which is an output leaf at depth at most $m$. As a binary tree has at most $2^m$ such leaves and each support has at most $(q-1)^m$ labelings (since
$|S_\gamma|\le\vol(S_\gamma)$), we have that
\[
 \sum_{\gamma:S_\gamma\ni v}
 w_h(\gamma)e^{2\vol(S_\gamma)}
 \le \sum_{m\ge1}
       [2(q-1)e^{-(\vartheta-2)}]^m
 \le \sum_{m\ge1}3^{-m}
 =\frac12.
\]

Every polymer incompatible with $\gamma$ contains a
vertex in the closed neighborhood $N[S_\gamma]$.
Since $|N[S_\gamma]|\le2\vol(S_\gamma)$, we obtain
\begin{align*}
 \sum_{\gamma'\not\sim\gamma}
 w_h(\gamma')e^{2\vol(S_{\gamma'})}
\le \sum_{v\in N[S_\gamma]}
       \sum_{\gamma':S_{\gamma'}\ni v}
       w_h(\gamma')e^{2\vol(S_{\gamma'})}\le \vol(S_\gamma).
\end{align*}
verifying the Koteck\'y--Preiss condition with
$\mathfrak f(\gamma)=\mathfrak g(\gamma)=\vol(S_\gamma)$.

For the vertex bound, apply Lemma~\ref{lem:marked-KP}
to the polymer model around $h$ with
$\mathfrak f(\gamma)=\mathfrak g(\gamma)=\vol(S_\gamma)$.
Since $\mathfrak g(\boldsymbol\gamma)=\vol(\boldsymbol\gamma)$,
the lemma gives
\[
 \sum_{\boldsymbol\gamma:\,v\in U(\boldsymbol\gamma)}
 |\Phi_h(\boldsymbol\gamma)|e^{\vol(\boldsymbol\gamma)}
 \le \sum_{\gamma:S_\gamma\ni v}
       w_h(\gamma)e^{2\vol(S_\gamma)}
 \le \frac12,
\]
where the last inequality is the bound established above. This proves \eqref{eq:volume-root-bound}.
\end{proof}

The following lemma gives efficient counting and sampling algorithms for this polymer model.

\begin{lemma}\label{lem:gas-algorithms}
Suppose $\vartheta\ge2+\log(6(q-1))$. For every stable reference $h$ and every $0<\xi<1/2$, the polymer model $\Xi_h$ admits a deterministic $\xi$-approximate counting algorithm and a $\xi$-approximate sampling algorithm, each with runtime $\operatorname{poly}(n,1/\xi)$.
\end{lemma}

\begin{proof}
Fix a stable reference $h$. We first construct the counting
algorithm by grouping cluster contributions according to
the union of their supports.

For $\Lambda\subseteq V$, let $\Xi_{\Lambda}(x)$ be the partition function
restricted to polymers supported in $\Lambda$, with each weight
replaced by $w_h(\gamma)x^{\vol(S_\gamma)}$.
Thus $\Xi_V(1)=\Xi_h$, and the exponent of $x$ records the total degree volume. For every connected set $U$, we have the
formal power series identity
\begin{equation}\label{volume:eqtag:6}
\begin{aligned}
 F_U(x)
 &:=\sum_{\Lambda\subseteq U}(-1)^{|U|-|\Lambda|}\log\Xi_{\Lambda}(x)\\
 &=\sum_{\boldsymbol\gamma:\,U(\boldsymbol\gamma)=U}
   \Phi_h(\boldsymbol\gamma)x^{\vol(\boldsymbol\gamma)}.
\end{aligned}
\end{equation}
Indeed, a cluster with union $T\subseteq U$ contributes
to $\log\Xi_{\Lambda}(x)$ exactly when $T\subseteq \Lambda$.
Its multiplier in the alternating sum is therefore
$\sum_{T\subseteq \Lambda\subseteq U}(-1)^{|U|-|\Lambda|}$,
which equals one if $T=U$ and zero otherwise.

Set $L=\lceil\log(100n/\xi)\rceil$.
Every cluster of total volume at most $L$ has a connected
union of volume at most $L$.
Consequently, summing the degree-$L$ truncations of
$F_U(x)$ over connected sets $U$ with $\vol(U)\le L$
gives the degree-$L$ truncation $F(x)$ of $\log\Xi_V(x)$.
By \eqref{eq:volume-root-bound},
\[
 \bigl|\log\Xi_h-F(1)\bigr|
 \le \sum_{\boldsymbol\gamma:\,\vol(\boldsymbol\gamma)>L}
          |\Phi_h(\boldsymbol\gamma)|
 \le \frac n2e^{-L}
 \le \frac{\xi}{200}.
\]
Since $0<\xi<1/2$, exponentiating this bound gives
\[
    (1-\xi)\Xi_h\le\exp(F(1))\le(1+\xi)\Xi_h.
\]

It remains to show that $F$ can be computed efficiently. Note that Lemma~\ref{volume:lem:encode} provides an algorithm that, given a vertex $v$ and a volume bound $L$, lists all connected sets $U$ containing $v$ with $\vol(U)\le L$. There are at most $2^{L+1}$ such sets for each $v$, and the enumeration takes $2^{L+O(1)}\operatorname{poly}(n)$
time. Running this algorithm for every $v\in V$ and
removing duplicates therefore lists all connected sets
of volume at most $L$, of which there are at most
$n2^{L+1}$.

For each enumerated set $U$ and each subset $\Lambda\subseteq U$, we compute $\Xi_{\Lambda}(x)$ by enumerating all spin assignments $\sigma$ on $\Lambda$ (extended by $h$ outside $\Lambda$). Let
$D=\{v\in \Lambda:\sigma_v\ne h_v\}$. Note that because the components of $D$ are disjoint and nonadjacent, we can efficiently check if a given configuration is valid. We can therefore obtain $\Xi_{\Lambda}(x)$ by summing the product of their component polymer weights
multiplied by $x^{\vol(D)}$.

Since $\Xi_{\Lambda}(0)=1$, the terms of $\log\Xi_{\Lambda}(x)$ through
degree $L$ are obtained by expanding
\[
    \sum_{j=1}^{L}\frac{(-1)^{j+1}}{j}
        \bigl(\Xi_{\Lambda}(x)-1\bigr)^j
\]
and discarding terms of degree greater than $L$. Combining these truncated logarithms using
\eqref{volume:eqtag:6} gives the degree-$L$ truncation of
$F_U(x)$; summing over the enumerated sets $U$ gives $F$.

For each $U$, the total number of assignments examined
across all subsets $\Lambda\subseteq U$ is
\[
    \sum_{\Lambda\subseteq U}q^{|\Lambda|}
    =(q+1)^{|U|}
    \le(q+1)^L.
\]
and so the total running time is
$n2^{L+1}(q+1)^L\operatorname{poly}(n,L) =\exp(O_q(L))\operatorname{poly}(n).$

We will now explain how to sample. As before, we will truncate as to discard polymers of volume greater than $L_0=\lceil\log(20n/\xi)\rceil$. Because a polymer occurs with probability at most its weight. Applying \eqref{eq:volume-root-bound} gives
\[
 \Pr_{\nu_h}\bigl(
   \text{some polymer of volume greater than }L_0
   \text{ occurs}\bigr)
 \le \sum_{\gamma:\,\vol(S_\gamma)>L_0}w_h(\gamma)
 \le \frac n2e^{-L_0}
 \le \frac{\xi}{10}
\]
and so this restriction changes the distribution by at most
$\xi/10$ in total variation.

We sample from the truncated model by deciding, one polymer
at a time, whether to include it. First, enumerate all
connected supports of volume at most $L_0$ using Lemma~\ref{volume:lem:encode}. For each support, we enumerate
the spin assignments that differ from $h$ at every vertex. After removing
duplicates, this gives a list $\mathcal P_0$ of size
\[
    m\le n2^{L_0+1}(q-1)^{L_0}
      =n\exp(O_q(L_0)).
\]

To sample from the truncated model, set
$\xi_{\rm step}=\xi/[100(m+1)]$. For
$\mathcal A\subseteq\mathcal P_0$, write $\Xi(\mathcal A)$
for its polymer partition function and
$\widehat\Xi(\mathcal A)$ for a $\xi_{\rm step}$-relative
approximation. Note that we can obtain this approximation by applying
the same counting procedure above with the added constraint of keeping an assignment only when every disagreement component defines a polymer in $\mathcal A$.

The sampler proceeds as follows.
\begin{enumerate}
    \item Initialize $\Gamma=\varnothing$ and
    $\mathcal P=\mathcal P_0$, and fix an ordering of
    $\mathcal P_0$.

    \item If $\mathcal P=\varnothing$, return $\Gamma$.
    Otherwise, take its first polymer $\gamma$ and set
    \[
        \mathcal P_\gamma
        =\{\gamma'\in\mathcal P\setminus\{\gamma\}:
              \gamma'\sim\gamma\}.
    \]

    \item Compute
    \[
        \widehat p_\gamma
        =\min\left\{1,\,
          w_h(\gamma)
          \frac{\widehat\Xi(\mathcal P_\gamma)}
               {\widehat\Xi(\mathcal P)}\right\}.
    \]
    With probability $\widehat p_\gamma$, add $\gamma$
    to $\Gamma$ and replace $\mathcal P$ by
    $\mathcal P_\gamma$. Otherwise, delete $\gamma$
    from $\mathcal P$. Return to Step~2.
\end{enumerate}

For correctness, fix a decision history with selected
collection $\Gamma$ and remaining family $\mathcal P$.
The compatible collections consistent with this history
are precisely $\Gamma\cup\Gamma'$, where
$\Gamma'\subseteq\mathcal P$ is compatible. Those containing
$\gamma$ correspond to compatible collections in
$\mathcal P_\gamma$. Consequently, under the truncated
measure, the conditional probability of including
$\gamma$ is
\[
    p_\gamma
    =
    \frac{\bigl(\prod_{\rho\in\Gamma}w_h(\rho)\bigr)
          w_h(\gamma)\Xi(\mathcal P_\gamma)}
         {\bigl(\prod_{\rho\in\Gamma}w_h(\rho)\bigr)
          \Xi(\mathcal P)}
    =
    w_h(\gamma)\frac{\Xi(\mathcal P_\gamma)}
                     {\Xi(\mathcal P)}.
\]
Since both partition-function estimates have relative
error at most $\xi_{\rm step}$, and clipping to $[0,1]$
cannot increase the error,
\[
    |\widehat p_\gamma-p_\gamma|
    \le \frac{2\xi_{\rm step}}{1-\xi_{\rm step}}p_\gamma
    \le 4\xi_{\rm step}.
\]

Couple the sampler's decisions with the target conditional
decisions using a common uniform random number whenever
their histories agree. At each such step, the probability
of disagreement is
$|\widehat p_\gamma-p_\gamma|\le4\xi_{\rm step}$.
Each iteration removes at least one polymer, so there
are at most $m$ decisions. A union bound therefore gives
total variation error at most $4m\xi_{\rm step}$ relative
to the truncated model. Including the initial truncation
error, the distance from $\nu_h$ is at most
\[
    \frac{\xi}{10}+4m\xi_{\rm step}
    \le \frac{\xi}{10}+\frac{\xi}{25}
    <\xi.
\]
as desired.
\end{proof}

\subsection{\texorpdfstring{$O$}{O}: Overlap Correction and Approximating the Measure}\label{sec:c3po-framework}
For any cut budget $b$, the preceding subsections construct
$\mathcal H_b$ and provide approximate counting and sampling algorithms
for the polymer model around each reference $h\in\mathcal H_b$.
We now choose $b$ so that these models represent almost all Gibbs mass, then correct for configurations represented more than once.

We first compare the partition functions $\Xi_h$ using
Definition~\ref{def:local-symmetry}. For notation, let
$\boldsymbol i=i\mathbf1_n$ and recall the constant
reference $z_0$ from that definition. Every constant reference is feasible
by Condition~\ref{cond:1} and stable because it has no cut edges.
Set
\[
    \delta=\min_{v\in V}d_v,\qquad E_R=2ne^{-\delta R},\qquad C_{\rm press}=q+1+B/2,
\]
with the convention that $E_\infty=0$.

\begin{lemma}\label{lem:pressure}
Under the hypotheses of Lemma~\ref{lem:boosted-KP} and
Definition~\ref{def:local-symmetry}, every polymer partition function around a stable reference $h$ satisfies
\begin{equation}\label{eq:surface-bound}
    |\log\Xi_h-\log\Xi_{z_0}|
    \le C_{\rm press}c(h)+E_R.
\end{equation}
\end{lemma}
\begin{proof}
Given a cluster $\boldsymbol\gamma$, let $n_v(\boldsymbol\gamma)$ count the occurrences of $v$ in the
cluster supports, and set
\[
    \alpha_v(\boldsymbol\gamma)
    =\frac{n_v(\boldsymbol\gamma)}{\sum_i|S_{\gamma_i}|},
    \qquad
    p_v(h)=\sum_{\boldsymbol\gamma}
           \Phi_h(\boldsymbol\gamma)\alpha_v(\boldsymbol\gamma),
\]
namely the part of $\log \Xi_h$ assigned to vertex $v$. Note that absolute convergence
and \eqref{eq:volume-root-bound} give the following:
\begin{align}
    \sum_v p_v(h)
    = \sum_{\boldsymbol\gamma}\Phi_h(\boldsymbol\gamma)
      \underbrace{\sum_v\alpha_v(\boldsymbol\gamma)}_{=1}
    = \log\Xi_h,\\
    |p_v(h)|
    \le \sum_{\boldsymbol\gamma:\,v\in U(\boldsymbol\gamma)}
        |\Phi_h(\boldsymbol\gamma)|
    \le \sum_{\boldsymbol\gamma:\,v\in U(\boldsymbol\gamma)}
        |\Phi_h(\boldsymbol\gamma)|e^{\vol(\boldsymbol\gamma)}
    \le \frac12.
\end{align}

Now, let $D_h$ be the endpoints of the cut edges (non-monochromatic edges) of $h$ so that
$|D_h|\le2c(h)$. Our strategy will be to compare $\log\Xi_h=\sum_vp_v(h)$ with
\[
    \sum_vp_v(\boldsymbol{h_v})
    =\sum_{i=1}^q\sum_{v:h_v=i}p_v(\boldsymbol i),
\]
where $\boldsymbol{h_v}=h_v\mathbf1_n$ is the constant reference
with spin $h_v$. We will then compare this sum with $\log\Xi_{z_0}$
and apply the triangle inequality.

Fix a cluster of polymers around $h$ whose union $U$ is disjoint
from $D_h$. Then $U$ is connected and $h$ is constant with spin $i$ on $U\cup N(U)$.
We compare the polymers around $h$ with those around the constant
reference $\boldsymbol i$. For each $\gamma=(T,a)$ in this cluster,
apply the same recoloring $a$ on $T$ to either reference.
The resulting configurations
$h^\gamma$ and $(\boldsymbol i)^\gamma$ agree on all edges incident
to $T$. Every other edge keeps its spins from the corresponding
reference, which is feasible: $h$ by assumption and $\boldsymbol i$
by Condition~\ref{cond:1}. Since the constraints are pairwise,
\[
    h^\gamma\in\mathcal F_G
    \quad\Longleftrightarrow\quad
    (\boldsymbol i)^\gamma\in\mathcal F_G.
\]
The references $h$ and $\boldsymbol i$ also agree on all edges
incident to $T$, so the recoloring changes their cut sizes by the
same amount:
\[
    \Delta_h(\gamma)=\Delta_{\boldsymbol i}(\gamma),
    \qquad w_h(\gamma)=w_{\boldsymbol i}(\gamma).
\]
Since $T$ has the same volume, \eqref{eq:overview-polymer-bound}
allows $\gamma$ around $h$ exactly when it allows $\gamma$ around
$\boldsymbol i$. The supports and incompatibility graphs are
unchanged, so \eqref{eq:ursell-definition} gives
$\Phi_h(\boldsymbol\gamma)=\Phi_{\boldsymbol i}(\boldsymbol\gamma)$.
Since
\[
    \sum_{v:h_v=j}\alpha_v(\boldsymbol\gamma)
    =\mathbf1\{j=i\},
\]
these clusters contribute equally to $\log\Xi_h$ and
$\sum_vp_v(\boldsymbol{h_v})$.

Only clusters meeting $D_h$ remain in the difference.
Applying \eqref{eq:volume-root-bound} to $h$ and the $q$ constant
references therefore gives
\begin{equation}\label{eq:reference-constant-comparison}
\begin{aligned}
    \left|\log\Xi_h-\sum_vp_v(\boldsymbol{h_v})\right|
    &\le \sum_{v\in D_h}\left(
       \sum_{\boldsymbol\gamma:\,v\in U(\boldsymbol\gamma)}
          |\Phi_h(\boldsymbol\gamma)|
       +\sum_{i=1}^q
        \sum_{\boldsymbol\gamma:\,v\in U(\boldsymbol\gamma)}
          |\Phi_{\boldsymbol i}(\boldsymbol\gamma)|
       \right)\\
    &\le \frac{q+1}{2}|D_h|
     \le(q+1)c(h).
\end{aligned}
\end{equation}

We now compare the contributions from constant references.
For each occurring type $\tau$, choose $(w,j)$ with $t(w,j)=\tau$
and set $p_\tau=p_w(\boldsymbol j)$, so $|p_\tau|\le1/2$.
Fix $(v,i)$ of the same type. For $1\le R<\infty$, take the
isomorphism $\pi$ and spin permutation $\sigma$ from
Condition~\ref{cond:3a}, with $\pi(v)=w$ and $\sigma(i)=j$.
A cluster containing $v$ with $\sum_\ell|S_{\gamma_\ell}|\le R$
has connected union contained in $B_G(v,R-1)$.
Thus every edge incident to one of its supports lies in $B_G(v,R)$.

Map each polymer in the cluster by
\[
    \gamma=(T,a)\longmapsto\gamma'=(\pi(T),a'),
    \qquad a'_{\pi(x)}=\sigma(a_x).
\]
For $z=(\boldsymbol i)^\gamma$ and
$z'=(\boldsymbol j)^{\gamma'}$, we have
$z'_{\pi(x)}=\sigma(z_x)$ throughout $B_G(v,R)$.
Condition~\ref{cond:3a} preserves the constraints on all affected
edges, while Condition~\ref{cond:1} makes the unchanged constant
references feasible. Hence $z$ is feasible if and only if $z'$ is.
Condition~\ref{cond:3a} also preserves degrees, and $\sigma$
preserves whether two spins differ, so it follows that
\begin{align}
    \Delta_{\boldsymbol j}(\gamma')
    &=\sum_{\substack{xy\in E\\\{x,y\}\cap T\ne\varnothing}}
       \mathbf1\{\sigma(z_x)\ne\sigma(z_y)\}
     =\Delta_{\boldsymbol i}(\gamma),\\
    \vol(\pi(T))&=\sum_{x\in T}d_{\pi(x)}=\vol(T),\\
    w_{\boldsymbol j}(\gamma')&=w_{\boldsymbol i}(\gamma).
\end{align}
Together with \eqref{eq:overview-polymer-bound}, these equalities show
that $\gamma$ is a polymer around $\boldsymbol i$ exactly when
$\gamma'$ is a polymer around $\boldsymbol j$.

Applying this map to each entry of the cluster preserves support
sizes, the incompatibility graph, and the number of occurrences of the root.
The inverse is given by $\pi^{-1}$ and $\sigma^{-1}$, so this is
a bijection between the clusters containing $v$ and $w$ of total
support size at most $R$. By \eqref{eq:ursell-definition},
\[
    \Phi_{\boldsymbol j}(\boldsymbol\gamma')
       \alpha_w(\boldsymbol\gamma')
    =\Phi_{\boldsymbol i}(\boldsymbol\gamma)
       \alpha_v(\boldsymbol\gamma).
\]
Their contributions to $p_v(\boldsymbol i)$ and $p_w(\boldsymbol j)$
therefore agree. Every omitted cluster has
$\vol(\boldsymbol\gamma)\ge\delta(R+1)$, so
\eqref{eq:volume-root-bound} gives, for either $(u,f)=(v,\boldsymbol i)$
or $(w,\boldsymbol j)$,
\[
    \sum_{\substack{\boldsymbol\gamma:\,u\in U(\boldsymbol\gamma)\\
                    \sum_\ell|S_{\gamma_\ell}|>R}}
       |\Phi_f(\boldsymbol\gamma)|\alpha_u(\boldsymbol\gamma)
    \le e^{-\delta(R+1)}
       \sum_{\boldsymbol\gamma:\,u\in U(\boldsymbol\gamma)}
          |\Phi_f(\boldsymbol\gamma)|e^{\vol(\boldsymbol\gamma)}
    \le \frac12e^{-\delta(R+1)}.
\]
It follows that
\begin{equation}\label{eq:pressure-type-comparison}
    |p_v(\boldsymbol i)-p_{t(v,i)}|\le e^{-\delta R}.
\end{equation}
For $R=0$, this follows from $|p_v(\boldsymbol i)|,|p_{t(v,i)}|\le1/2$.
For $R=\infty$, the isomorphism is defined on the entire connected
component and applies to every cluster, giving equality of the
two contributions.

By Conditions~\ref{cond:3b} and~\ref{cond:3c}, respectively,
\begin{align*}
    \sum_\tau|n_\tau(h)-n_\tau(z_0)|
    &=\sum_\tau\left|\sum_{uv\in E}a_{uv,\tau}(h_u,h_v)\right|\\
    &\le\sum_{uv\in E}\sum_\tau|a_{uv,\tau}(h_u,h_v)|
     \le Bc(h).
\end{align*}
Applying \eqref{eq:pressure-type-comparison} to the spins of $h$
and $z_0$ incurs total error at most $2ne^{-\delta R}=E_R$.
Since $z_0$ is constant and $|p_\tau|\le1/2$, we obtain
\begin{align*}
    \left|\sum_vp_v(\boldsymbol{h_v})-\log\Xi_{z_0}\right|
    &\le E_R+
       \left|\sum_\tau[n_\tau(h)-n_\tau(z_0)]p_\tau\right|\\
    &\le E_R+\frac12\sum_\tau|n_\tau(h)-n_\tau(z_0)|
     \le E_R+\frac B2c(h).
\end{align*}
Combining this with \eqref{eq:reference-constant-comparison} gives
\[
    |\log\Xi_h-\log\Xi_{z_0}|
    \le(q+1)c(h)+\frac B2c(h)+E_R
    =C_{\rm press}c(h)+E_R
\]
as desired.
\end{proof}

We now show that, for sufficiently large $\theta$, we can choose
the cut budget $b$ so that the polymer models around $\mathcal H_b$
represent almost all Gibbs mass. We use \eqref{eq:surface-bound} and
Lemma~\ref{potts:lem:enumeration} to bound the mass of configurations
whose reference has cut greater than $b$.

For notation, recall that $w(z)=e^{-\theta c(z)}$ and that $\mathcal U_h$ is the
set of configurations represented by the polymer model around $h$.
For each $z\in\mathcal U_h$, its representation is uniquely determined
by the connected components of $\{v:z_v\ne h_v\}$ and the spins of
$z$ on them. Summing \eqref{eq:overview-representation-weight}
over compatible polymer collections therefore gives
\begin{equation}\label{eq:reference-image-weight}
    \sum_{z\in\mathcal U_h}w(z)=w(h)\Xi_h.
\end{equation}
With $A_{\rm vol}$ as in Lemma~\ref{potts:lem:enumeration}, set
\begin{equation}\label{eq:structural-penalty}
 \theta_0=\max\left\{
 \frac{16(q-1)}{\eta}[2+\log(6(q-1))],\quad
 A_{\rm vol}+C_{\rm press}+2\right\}.
\end{equation}

\begin{lemma}\label{lem:c3po-tail}\label{cor:reference-cutoff}
Suppose $\theta\ge\theta_0$. For $0<\varepsilon<1/2$, choose
\begin{equation}\label{eq:proof-reference-cutoff}
    b=\left\lceil
       \frac{\log M+E_R+\log(100/\varepsilon)}2
      \right\rceil = O(k_\eta+E_R+\log(n/\varepsilon)),
\end{equation}
where $M$ is given by Lemma~\ref{potts:lem:enumeration}.
Then $\mu\{z:c(\kappa(z))>b\}\le\varepsilon/100$.
\end{lemma}
\begin{proof}
By Lemma~\ref{lem:c3po-canonical}, every configuration is represented
around its reference under the map $\kappa$, so $\{z:\kappa(z)=h\}\subseteq\mathcal U_h$.
Since $c(z_0)=0$, \eqref{eq:reference-image-weight} gives
$Z_G(\theta)\ge\Xi_{z_0}$. Applying \eqref{eq:surface-bound} gives
\[
    \mu\{z:\kappa(z)=h\}
    \le\frac{e^{-\theta c(h)}\Xi_h}{Z_G(\theta)}
    \le e^{E_R-(\theta-C_{\rm press})c(h)}.
\]
There are at most $|\mathcal H_j|\le Me^{A_{\rm vol}j}$ stable
references with cut $j$. Writing $s=\theta-A_{\rm vol}-C_{\rm press}
\ge2$ and summing over these references gives
\begin{equation}\label{eq:c3po-tail}
\begin{aligned}
    \mu\{z:c(\kappa(z))>b\} \le Me^{E_R}\sum_{j=b+1}^{\infty}e^{-sj}=Me^{E_R}\frac{e^{-s(b+1)}}{1-e^{-s}}\le\frac{\varepsilon}{100}\frac{e^{-2}}{1-e^{-2}}
     <\frac{\varepsilon}{100}.
\end{aligned}
\end{equation}
where the last line uses \eqref{eq:proof-reference-cutoff} and $s\ge2$.
\end{proof}

Now, let $\mathcal U=\bigcup_{h\in\mathcal H_b}\mathcal U_h$. Note that even though sampling from $\mu(\cdot\mid\mathcal U)$ would suffice, since
\eqref{eq:c3po-tail} gives
\[
    \TV\bigl(\mu(\cdot\mid\mathcal U),\mu\bigr)
    =\mu(\mathcal U^c)\le\varepsilon/100,
\]
this does not give us an algorithm. The issue is that the same configuration $z$ may be represented around several different references $h$. If we choose $h$ with
probability proportional to $w(h)\Xi_h$ and sample its polymer
model, then, by \eqref{eq:overview-representation-weight}, the probability of proposing $z$ is therefore proportional to $w(z)$ times the number of references representing it.

To get around this, we fix an ordering of the references and accept a proposal only
when the chosen $h$ is the first reference representing $z$. This retains exactly one contribution of weight $w(z)$ per
represented configuration.

\begin{lemma}\label{lem:c3po-overlap}
Let $\mathcal S=(h_1,\ldots,h_K)$ be a nonempty list of distinct references,
and set
\[
    W_{\mathcal S}=\sum_{h\in \mathcal S}w(h)\Xi_h,\qquad
    \mathcal U_{\mathcal S}=\bigcup_{h\in \mathcal S}\mathcal U_h.
\]
Consider the distribution on pairs $(h,z)$ induced by the following
sampling procedure:
\begin{enumerate}
    \item Choose $h\in \mathcal S$ with probability $w(h)\Xi_h/W_{\mathcal S}$.
    \item Sample $\Gamma\sim\nu_h$ and apply $\Gamma$ to $h$
    to obtain $z$.
    \item Accept if $h$ is the first reference in $\mathcal S$ representing $z$.
\end{enumerate}
Then, conditional on acceptance, $z$ has distribution
$\mu(\cdot\mid\mathcal U_{\mathcal S})$, and the acceptance probability satisfies
\begin{equation}\label{eq:c3po-acceptance}
    p_{\mathcal S}=\frac{Z_G(\theta)\mu(\mathcal U_{\mathcal S})}{W_{\mathcal S}}\ge\frac1K.
\end{equation}
\end{lemma}
\begin{proof}
For each represented pair $(h,z)$, let $\Gamma$ be its unique
polymer collection. By \eqref{eq:overview-representation-weight},
the probability of proposing this pair is
\[
    \frac{w(h)\Xi_h}{W_{\mathcal S}}\,
    \frac{\prod_{\gamma\in\Gamma}w_h(\gamma)}{\Xi_h}
    =\frac{w(z)}{W_{\mathcal S}}.
\]
Exactly one pair is accepted for each $z\in\mathcal U_{\mathcal S}$, so
\[
    p_{\mathcal S}=\sum_{z\in\mathcal U_{\mathcal S}}\frac{w(z)}{W_{\mathcal S}}
       =\frac{Z_G(\theta)\mu(\mathcal U_{\mathcal S})}{W_{\mathcal S}}.
\]
Conditioning on acceptance gives $\mu(\cdot\mid\mathcal U_{\mathcal S})$.
If $m_{\mathcal S}(z)$ is the number of references representing $z$, then
\[
    W_{\mathcal S}=\sum_{z\in\mathcal U_{\mathcal S}}m_{\mathcal S}(z)w(z)
    \le K\sum_{z\in\mathcal U_{\mathcal S}}w(z)
    =KZ_G(\theta)\mu(\mathcal U_{\mathcal S}),
\]
which proves $p_{\mathcal S}\ge1/K$.
\end{proof}

We now combine these ingredients to give the counting and sampling
algorithms in Theorem~\ref{thm:structural-rank}.

\begin{algorithm}[H]
\caption{Approximate counting and sampling}
\label{alg:c3po-counting-sampling}
\begin{algorithmic}[1]
\Require A graph $G$ and model as in Theorem~\ref{thm:structural-rank},
the radius $R$, a fixed $\theta\ge\theta_0$, and $0<\varepsilon<1/2$.
\Ensure An $\varepsilon$-approximate sample from $\mu$, or an
$\varepsilon$-approximation to $Z_G(\theta)$.
\State Choose $b$ as in \eqref{eq:proof-reference-cutoff}.
\State Run Algorithms~\ref{alg:approximate-references}
and~\ref{alg:stable-references} with $\zeta=\varepsilon/100$.
\State Append $z_0$ to the output, remove duplicates, and arrange
the references in a fixed order to obtain $\mathcal S$.
\State Set $K\gets|\mathcal S|$ and $\xi\gets\varepsilon/(10000K)$.
\State Use the deterministic counting algorithm of
Lemma~\ref{lem:gas-algorithms} to obtain $\xi$-approximations
$\widehat\Xi_h$ for all $h\in \mathcal S$, and set
$\widehat W_{\mathcal S}\gets\sum_{h\in \mathcal S}w(h)\widehat\Xi_h$.
\State \textbf{For each trial:} using fresh randomness, choose $h\in \mathcal S$
with probability $w(h)\widehat\Xi_h/\widehat W_{\mathcal S}$ and use
Lemma~\ref{lem:gas-algorithms} to draw a $\xi$-approximate sample
$\Gamma$ from $\nu_h$. Apply $\Gamma$ to $h$ to obtain $z$.
Accept if $h$ is the first reference in $\mathcal S$ representing $z$.
\State \textbf{For sampling:} run at most
$T=\lceil4K\log(100/\varepsilon)\rceil$ trials and return the first
accepted configuration. If every trial is rejected, return $z_0$.
\State \textbf{For counting:} run $J=\lceil4000K/\varepsilon^2\rceil$
trials and let $\overline p$ be the fraction accepted.
Return $\widehat Z=\widehat W_{\mathcal S}\overline p$.
\end{algorithmic}
\end{algorithm}

\begin{proof}[Proof of Theorem~\ref{thm:structural-rank}]
Take $\theta_0$ as in \eqref{eq:structural-penalty} and fix
$\theta\ge\theta_0$. The local symmetry assumptions give
\eqref{eq:surface-bound}, and the choice of $\theta_0$ ensures
that Lemmas~\ref{lem:gas-algorithms} and~\ref{lem:c3po-tail} apply.

By Lemma~\ref{potts:lem:enumeration}, the list $\mathcal S$ is contained in
$\mathcal H_b$ and contains every reference in $\mathcal H_b$ with
probability at least $1-\varepsilon/100$. The constant reference
$z_0$ is feasible by Condition~\ref{cond:1} and has no cut edges,
so appending it ensures that $\mathcal S$ is nonempty.
When the list is complete, Lemma~\ref{lem:c3po-canonical}
and \eqref{eq:c3po-tail} give
\[
    \varepsilon_{\rm miss}:=1-\mu(\mathcal U_{\mathcal S})
    \le\mu\{z:c(\kappa(z))>b\}\le\varepsilon/100.
\]

Condition on the ordered list $\mathcal S$. The estimates $\widehat\Xi_h$
are then deterministic, and the trials are independent.
Let $\rho$ be the distribution of the pair $(h,z)$ proposed in
Lemma~\ref{lem:c3po-overlap}, and let $\widehat\rho$ be the
distribution of the pair proposed in one trial of
Algorithm~\ref{alg:c3po-counting-sampling}.
Each $\widehat\Xi_h$ has relative error at most $\xi$, so the same
holds for $\widehat W_{\mathcal S}$. Combining the error in choosing $h$
with the error in sampling from $\nu_h$ gives
\begin{align*}
    \TV(\rho,\widehat\rho)
    &\le\frac12\sum_{h\in \mathcal S}
       \left|\frac{w(h)\widehat\Xi_h}{\widehat W_{\mathcal S}}
                   -\frac{w(h)\Xi_h}{W_{\mathcal S}}\right|+\xi\le\frac{\xi}{1-\xi}+\xi\le3\xi.
\end{align*}
The acceptance rule is the same for both procedures.
Writing $p$ and $\widehat p$ for their acceptance probabilities,
\eqref{eq:c3po-acceptance} gives
\[
    |\widehat p-p|\le3\xi,\qquad
    p\ge1/K,\qquad \widehat p\ge1/(2K).
\]

We first verify the sampling guarantee for a fixed complete list $\mathcal S$.
Conditioning on acceptance gives
\[
    \TV\bigl(\rho(\cdot\mid\mathrm{accept}),
             \widehat\rho(\cdot\mid\mathrm{accept})\bigr)
    \le\frac{2\TV(\rho,\widehat\rho)}{\widehat p}
    \le12K\xi.
\]
By Lemma~\ref{lem:c3po-overlap}, the distribution of $z$ under
$\rho$, conditional on acceptance, is $\mu(\cdot\mid\mathcal U_{\mathcal S})$.
The same error bound holds after discarding $h$.
By independence, conditional on at least one acceptance, the first
accepted configuration has the law of a single trial conditioned
on acceptance.
The probability that all $T$ trials are rejected satisfies
\[
    (1-\widehat p)^T\le e^{-T/(2K)}
    \le(\varepsilon/100)^2\le\varepsilon/100.
\]
Also, $\TV(\mu(\cdot\mid\mathcal U_{\mathcal S}),\mu)=\varepsilon_{\rm miss}\le\varepsilon/100$.
Adding these errors and the probability that the list is incomplete
gives total variation error at most $12K\xi+3\varepsilon/100<\varepsilon$.

For counting, again fix a complete list $\mathcal S$.
Equation~\eqref{eq:c3po-acceptance} gives
$W_{\mathcal S}p=Z_G(\theta)(1-\varepsilon_{\rm miss})$, while the preceding estimates give
\[
    \left|\frac{\widehat W_{\mathcal S}}{W_{\mathcal S}}-1\right|\le\xi,
    \qquad
    \left|\frac{\widehat p}{p}-1\right|\le3K\xi.
\]
Each trial is accepted independently with probability $\widehat p$.
Thus Chebyshev's inequality gives
\[
    \Pr\left(\left|\frac{\overline p}{\widehat p}-1\right|
             >\frac{\varepsilon}{10}\right)
    \le\frac{100}{J\varepsilon^2\widehat p}\le\frac1{20}.
\]
The output satisfies
\[
    \frac{\widehat Z}{Z_G(\theta)}
    =\frac{\widehat W_{\mathcal S}}{W_{\mathcal S}}
     \frac{\overline p}{\widehat p}
     \frac{\widehat p}{p}(1-\varepsilon_{\rm miss}).
\]
With conditional probability at least $19/20$, the first three factors lie in
$[1-\varepsilon/10,1+\varepsilon/10]$ and
$1-\varepsilon_{\rm miss}\in[1-\varepsilon/10,1]$. Consequently,
\[
    1-\varepsilon
    \le(1-\varepsilon/10)^4
    \le\frac{\widehat Z}{Z_G(\theta)}
    \le(1+\varepsilon/10)^3
    \le1+\varepsilon.
\]
Including the probability that the list is incomplete, the estimate
fails with probability at most $1/20+\varepsilon/100<1/4$.

It remains to bound the running time.
By \eqref{eq:proof-reference-cutoff} and
Lemma~\ref{potts:lem:enumeration}, constructing $\mathcal S$ takes
$\exp(O(k_\eta+E_R))\operatorname{poly}(n,1/\varepsilon)$ time,
and $K$ is at most this large.
Since $1/\xi=10000K/\varepsilon$, each call to
Lemma~\ref{lem:gas-algorithms} takes
$\operatorname{poly}(n,K,1/\varepsilon)$ time.
Each trial also checks at most $K$ references to determine whether
to accept $z$. For each reference $h$, this amounts to finding the
connected components of $\{v:z_v\ne h_v\}$ and checking
Definition~\ref{def:overview-polymer-model}, which takes polynomial time.
There are $K$ counting calls and at most $T$ or $J$ trials, so the
total running time is
\begin{equation}\label{eq:c3po-runtime}
    \exp(O(k_\eta+E_R))\operatorname{poly}(n,1/\varepsilon).
\end{equation}
Substituting $E_R=2ne^{-\delta R}$ gives the claimed bound, with
constants depending only on $q,\eta,B$ and the fixed $\theta$.
\end{proof}

\newpage
\section*{Acknowledgements}
The inspiration for this work originated from discussions between the author and Pedro Paredes on whether ideas from subspace enumeration could be used to tackle $\#\mathrm{BIS}$. The author is very thankful to Pedro for these (and many other) conversations. The author is also grateful to Will Perkins for helpful feedback on initial stages of this work and for many insightful conversations about $\#\mathrm{BIS}$ (and sampling in general) over the last year.

\paragraph{Statement of AI use.} All ideas originated from the author. AI was used primarily for computation and for testing whether proposed ideas and arguments worked. The project developed over more than 100 interactions with ChatGPT 5.6 Sol and, later, GPT-6 Astra.

The project began as a follow-up to~\cite{NP}, with the goal of obtaining subexponential algorithms for $\BIS$ in the spirit of those for Unique Games. Since~\cite{NP} provides algorithms on expanders, the initial idea was to go beyond $\sigma_2$ via subspace enumeration.

Early AI-assisted attempts pointed the author to the idea of utilizing active regression to recover cut references. However, the suggested algorithms all required various ad hoc, graph-dependent conditions. In examining these arguments, we identified symmetry as the missing ingredient needed to compare the weights of deviations around different references. With AI assistance, we then developed the vertex-transitivity based algorithm for $\BIS$. We next attempted to weaken or remove the symmetry assumption, leading to the large-girth results.

These developments clarified the role of symmetry in balancing the energetic cost of a reference cut against the entropy of deviations around it. This naturally suggested the ferromagnetic Potts model, which also penalizes cuts but has a stronger spin symmetry than the hard-core model. After extending the approach to this setting, we abstracted the common ingredients into a general framework, from which the $\BIS$ and Potts algorithms follow as applications of a single theorem. The author assumes responsibility for all content.

\bibliographystyle{alpha}
\bibliography{ref}
\end{document}